\documentclass[12pt]{article}

\title{\textcolor{black}{Loss aversion in strategy-proof school-choice mechanisms}\thanks{An older version was circulating under the title ``school choice and loss aversion." We thank Georgy Artemov, In\'acio B\'o, Rustamdjan Hakimov, Fabian Herweg, Peter Katu\v{s}\v{c}\'{a}k, Dorothea K\"ubler, Matthias Lang, Takeshi Murooka, Antonio Rosato, Maybritt Schillinger, Roland Strausz, Georg Weizs\"acker, and seminar participants in Aachen, Berlin, and Munich, as well as at CED'19, ESEM'19, MIP'19, VfS'19, AMETS, and CMID'20 for useful comments and suggestions. Financial support by Deutsche Forschungsgemeinschaft (CRC/TRR 190, project 280092119) and UniCredit Foundation (Modigliani Grant) is gratefully acknowledged.}}
\author{Vincent Meisner\thanks{Technical University Berlin, Stra\ss{}e des 17. Juni 135, 10623 Berlin, Germany, vincent.meisner@tu-berlin.de.} \and  Jonas von Wangenheim\thanks{University of Bonn, Institute for Microeconomics, Adenauerallee 24-42, 53113 Bonn, Germany, jwangenheim@uni-bonn.de.}}

\date{July 15, 2022}

\usepackage{xcolor}
\usepackage{colortbl}
\usepackage{booktabs}
\usepackage{bbm}
\usepackage{hhline}
\usepackage{amsmath}
\usepackage{array}
\usepackage{amssymb}
\usepackage{amsthm}
 \usepackage{arcs}
\usepackage[latin1]{inputenc}
\usepackage{parskip}
\usepackage{sgame}
\usepackage{color}
\usepackage{cancel}
\usepackage{hyperref}
\usepackage{wrapfig}
\usepackage{graphics}
\usepackage{graphicx}
\usepackage{tikz,pgfplots}
\usepackage{verbatim}
\usepackage[round]{natbib}
\usepackage{subfig}
\usepackage{gensymb}
\usepackage{pdfpages}
\usepackage[paper=a4paper,margin=1in,top=1in,bottom=1in]{geometry}
\usepackage[onehalfspacing]{setspace} 
\usepackage{thmtools, thm-restate}
\usepackage{cancel}
\usepackage{titlesec}
\usepackage{enumerate}
\usepackage{nicefrac}
\usepackage{adjustbox}
\usepackage{pdflscape}
\usepackage{rotating}
\usepackage[normalem]{ulem}

\titlespacing{\section}{0pt}{\parskip}{-\parskip+0.5em}
\titlespacing{\subsection}{0pt}{\parskip}{-\parskip+0.5em}
\titlespacing{\subsubsection}{0pt}{\parskip}{-\parskip+0.5em}
    \makeatletter
    \g@addto@macro\normalsize{\setlength\abovedisplayskip{6pt}}
    \g@addto@macro\normalsize{\setlength\belowdisplayskip{6pt}}
    \makeatother

\newtheoremstyle{slplain}
  {.5\baselineskip\@plus.2\baselineskip\@minus.2\baselineskip}
  {.5\baselineskip\@plus.2\baselineskip\@minus.2\baselineskip}
  {\slshape}
  {}
  {\bfseries}
  {.}
  { }
  {}

\theoremstyle{slplain}
\newtheorem{prop}{Proposition}
\newtheorem{lem}{Lemma}

\newtheorem{definition}{Definition}
\newtheorem{ass}{Assumption}
\newtheorem{cor}{Corollary}
\newtheorem{example}{Example}

\newcommand{\ea}[1]{\begin{align*}#1\end{align*}}
\newcommand{\eq}[1]{\begin{equation}#1\end{equation}}
\newcommand{\ean}[1]{\begin{align}#1\end{align}}

\allowdisplaybreaks

\usetikzlibrary{positioning,arrows,patterns,decorations.pathreplacing}
\tikzset{
  state/.style={circle,draw,minimum size=6ex},
  arrow/.style={-latex, shorten >=1ex, shorten <=1ex}}

\DeclareFontFamily{U}{mathx}{\hyphenchar\font45}
\DeclareFontShape{U}{mathx}{m}{n}{
      <5> <6> <7> <8> <9> <10>
      <10.95> <12> <14.4> <17.28> <20.74> <24.88>
      mathx10
      }{}
\DeclareSymbolFont{mathx}{U}{mathx}{m}{n}
\DeclareFontSubstitution{U}{mathx}{m}{n}
\DeclareMathAccent{\widecheck}{0}{mathx}{"71}
\DeclareMathAccent{\wideparen}{0}{mathx}{"75}

\pgfplotsset{soldot/.style={color=black,only marks,mark=*}} \pgfplotsset{holdot/.style={color=black,fill=white,only marks,mark=*}}

\begin{document}

\maketitle
\begin{abstract} \noindent
Evidence suggests that participants in \textcolor{black}{strategy-proof matching} mechanisms play dominated strategies. To explain the data, we introduce expectation-based loss aversion into a school-choice setting and characterize choice-acclimating personal equilibria. We find that non-truthful preference submissions can be strictly optimal if and only if they are top-\textcolor{black}{rank} monotone. In equilibrium, \textcolor{black}{inefficiency or justified envy} may \textcolor{black}{arise in seemingly stable or efficient mechanisms}. Specifically, students who are more loss averse or less confident than their peers \textcolor{black}{obtain suboptimal allocations.}
\end{abstract}
\noindent
JEL-Classification: C78, D47, D78, D81, D82, D91. 

\newpage


\section{Introduction} \label{sec:intro}


\textcolor{black}{Strategy-proof mechanisms} offer a celebrated solution to the problem of matching prospective students to schools.
\textcolor{black}{For instance, deferred-acceptance (DA) mechanisms implement stable allocations, and top-trading cycles (TTC) implement efficient allocations}. 
 Consequently, \textcolor{black}{such mechanisms are} used in many existing school choice programs.\footnote{For instance, \citet{pathak2013} provide many examples.}
By submitting their true preferences, students can maximize the probability of getting into their
most preferred school without hurting their chances of admission to other schools.
Unfortunately, growing evidence 
from both the field and the lab suggests
that (especially, but not only) students with low priority tend to conceal their preferences for popular schools
and mimic preferences for district schools despite
the dominance of the truthful \textcolor{black}{reporting}. Hence, potentially, none of the desired properties \textcolor{black}{like efficiency and stability} may be obtained.


We explain this puzzle with expectation-based loss aversion (EBLA, \cite{koszegi2006,koszegi2007}).
In our framework, the preference report is a channel to manipulate the expectations to which final match outcomes are compared.
Ranking a popular school behind a less preferred school 
is always costly in terms of the expected match utility, as such a rank-ordered list (ROL) shifts 
a part of the match probability
to an inferior school. However, it also mitigates disappointment, and
not even trying to get into the popular school by dropping it completely shields off potential disappointment \textcolor{black}{with respect to admisssion at this school}.
We characterize that ROLs are strictly rationalizable as a choice-acclimating equilibrium (CPE) in \textcolor{black}{static strategy-proof mechanisms} if and only if they satisfy a property we call top-\textcolor{black}{rank} monotonicity, which is a testable prediction.\footnote{\textcolor{black}{An ROL
is top-rank monotone if the rank of all schools preferred over the first-ranked school is decreasing in the true rank and 
the rank of the other schools is increasing in the true rank.} Such ROLs are indeed common in the data by \citet{li2017}.}
This theoretical foundation of commonly observed deviations is the first contribution of this paper.


\textcolor{black}{Secondly}, we show that these misrepresentations may give rise to justified envy and inefficiency in equilibrium. We  analyze choice-acclimating Bayesian Nash equilibria when heterogeneously loss-averse students
compete for scarce seats at elite schools. More specifically, loss-averse students decide to apply to their district schools
over the elite schools if they are pessimistic about their admission chances. Consequently, weaker students with a lower degree of loss aversion (or higher degree of confidence), who submit true preferences, are accepted instead.
In that sense, \textcolor{black}{strategy-proofness} does not ``level the playing field," voiding one of the crucial advantages prominently named by \citet{pathak2008}.
Our model also highlights a flaw in the empirical strategy to identify preferences reported to \textcolor{black}{strategy-proof mechanisms} as true.
Regarding affirmative action policy, this insight is important because the observation that certain students
do not rank certain schools does not necessarily mean that they prefer other schools.


In our model, students privately learn their match values for each school and
their individual degree of loss aversion.
Moreover, they receive a signal about their relative priorities compared to the other students
at each school.
Generally, given beliefs about the other students' priorities and strategies,
a student's preference report corresponds to a lottery over match outcomes.
For instance, by swapping two schools' ranks in the reported ROL, match probability mass
is shifted from one school to the other.
With respect to match utility alone, truthful reporting is a dominant strategy and,
thus, induces a lottery that first-order stochastically dominates any lottery
induced by any other ROL.
Following the CPE framework by \citet{koszegi2007}, the chosen outcome lottery constitutes the reference point.
In addition to match utility, students receive psychological utility from comparing an outcome to the reference point. Since losses with respect to the reference point are weighted stronger than gains, any uncertainty in the match utility distribution generates a cost in expected utility.


As \citet{koszegi2007} have already proved, CPE allows for a preference for
stochastically dominated lotteries, if an agent's loss aversion is sufficiently strong.
Indeed, a loss-averse student may prefer to be matched with school $x$ with certainty
over being matched with the same school $x$ with probability $(1-\epsilon)$ and being matched with
an even better school $y$ with probability $\epsilon>0$.
Intuitively, the mere possibility of getting into $y$ makes
the realization of the more likely outcome $x$ more painful.
Not listing $y$ abandons all hope so that this school does not enter the stochastic reference point
and disappointment is avoided.
Such motifs can explain the evidence, suggesting that
low- and mid-priority students are prone to misrepresentations, but high-priority and optimistic students are not.


We draw on the extensive literature on matching mechanisms, but
depart from the standard framework where preferences are only ordinal. In their seminal paper, \citet{gale1962} introduce
the deferred-acceptance mechanism as a solution to find the optimal stable matchings for the proposing side in the one-to-one matching problem. \textcolor{black}{The dominance of the truthful strategy for proposers in DA and TTC mechanisms was established by \citet{roth1982,roth1982ttc}.}
\citet{balinski1999} show that DA is constrained efficient in the sense that
no other fair mechanism Pareto-dominates it. Our model introduces a fundamentally different structure of incentives and questions
all of these classical insights. \citet{roth1989} and \citet{ehlers2007} are the first
to study matching with incomplete information.



\citet{hassidim2017note} gather stylized facts about the pervasive misrepresentation of preferences in truthful mechanisms. Similar to \citet{rees2018a} and \citet{chen2019} who analyze survey data, 
they find that ``misrepresentation rates are higher in
weaker segments of markets" 
and increase ``when applicants expect to face stronger competition," in line with the predictions of our model. 
In field data, misrepresentations are hard to identify since the true preferences are not observable.
However, \citet{hassidim2017ipmm}, \citet{shorrer2017} and \citet{artemov2020}
exploit objective rankings in their data
to expose ``obvious misrepresentations"
and find the same pattern.\footnote{They study the
Israeli Psychology Master's Match, Hungarian college admission, and Australian college admissions,
respectively. Naturally, all students should prefer a school with a scholarship over the same school
without scholarship, but the authors record that students forgo tuition waivers and no-strings-attached stipends.} 
\citet{artemov2020} and \citet{hassidim2017ipmm} find that 1\,-\,20\% 
and 2--8\% of obvious misrepresentations are ex-post costly, respectively. 
\citet{shorrer2017} estimate that the 12--19\% costly obvious misrepresentations
amount to \$3,000--\$3,500 on average (unconditionally \$347--\$738 per misrepresentation).
That is, even when restricting attention to obvious misrepresentations, consequential deviations can be observed.


Truthfulness is easier to detect in the lab where preferences are induced  by the experimental design.
While the pioneers \citet{chen2006} focused on a comparison of different mechanisms,
more recently researchers have investigated patterns in preference manipulations.
\citet{hakimov2020} provide a well-structured overview of the current state of experimental research on matching markets. They document that rates of truthfulness seem to depend on multiple factors which should not impede the dominance of the strategy
and which vary widely between studies. 
Rather than rooted in behavioral theory, most experimental studies are descriptive.
For instance, \citet{chen2006} introduced the district-school bias and the small-school bias, which
capture the tendency that safe district schools are ranked higher and small schools are ranked lower.
We offer a theory to explain this pattern.


In contrast to our paper, where students deliberately submit incorrect preferences, 
misrepresentations have most commonly been interpreted as cognitive failures to identify
the dominant strategy.\footnote{See, e.g., \citet{basteck2018}. 
However, when priorities and preferences are induced by the experimenter, 
we see the same individuals play a dominant or dominated strategy depending on their assigned score. 
\citet{hassidim2017ipmm} observe the same pattern in a high-ability population (compared to the general population). Controlling for cognitive ability, \citet{shorrer2017}
and \citet{artemov2020} find a causal relationship between admission selectivity and dominated choices.}
\citet{li2017} points out that DA is not ``obviously strategy-proof",\footnote{We discuss the differences between the two concepts in Section \ref{sec:OSPvsEBLA}.} and shows that most ``mistakes" vanish when replacing DA with sequential serial dictatorship.
We \textcolor{black}{find that} the most common deviations documented by \citet{li2017} \textcolor{black}{are indeed top-rank monotone}. 
Hence, our model of non-standard preferences provides an alternative explanation for the observations in \citet{li2017}. We believe that both explanations, cognitive mistakes and non-standard utility, are relevant in practice.\footnote{The persistence of misrepresentations even in high-stake environments, with trained participants and after many rounds of practice suggests that cognitive mistakes are not the only explanation. That some participants respond to advice and training suggests that our bias is not the only explanation. We see it as one piece of the puzzle. Even if the relevance of loss aversion in market-design practice was in doubt, we show to what extend loss aversion can confound experimental evidence and how it can be tested.}



Combinations of behavioral theory and matching are still relatively rare.
To the best of our knowledge, the first paper to consider non-standard preferences in matching
is by \citet{antler2015} whose agents' preferences are directly affected by the reported preferences of others.
\citet{fernandez2018deferred} studies anticipated regret in deferred acceptance, and \citet{zhang2021} considers school choice with level-k reasoning.
\citet{dreyfuss2019} recently and independently raised the point that EBLA can help explain misrepresentations in DA. Alongside various differences in modeling choices,
they focus on the individual decision problem and use empirical strategies to identify loss aversion in existing experimental data. In contrast, we take a deeper theoretical approach by deriving characterization results on rationalizable ROLs, analyzing strategic interaction, and evaluating remedy mechanisms. We discuss the distinction to our paper more carefully in Section \ref{sec:other-paper}. 
\textcolor{black}{\citet{meisner2022} proposes report-dependent utility as an explanation and differentiates our model from his by naming settings in which predictions of the models} differ.

\textcolor{black}{
\citet[p.~22]{gross2015} document that parents
who ``did not get what they hoped for and felt this sense of frustration and false hope" question the legitimacy of centralized school choice. However, they  they also conclude that ``only a small share of families probably experience the disappointment" because most applicants are matched to one of their top-ranked schools. To emphasize that disappointment already affects the submission of preferences, we employ the choice-acclimating personal equilibrium (CPE) concept introduced by \citet{koszegi2006}.}
It essentially captures disappointment aversion similar to \citet{bell1985}, \citet{loomes1986} or \citet{gul1991}, who model the reference point as the lottery's certainty equivalent. \textcolor{black}{\citet[Section~5]{handbook2018reference} compare these models in detail.} We choose CPE where outcomes are compared to the lottery's full distribution because it allows for ``mixed feelings" and because it is unclear what the certainty equivalent of a lottery over real school placements is supposed to be. \textcolor{black}{Comparing both approaches, \citet{sprenger2015endowment} finds more support for \citet{koszegi2006} in the data.}

EBLA is supported by evidence from the field, such as
\hbox{\citet{crawford2011}} 
or \citet{pope2011}. 
Evidence from the lab is mixed. While the conflicting evidence of \citet{ericson2011} and
\citet{heffetz2014} is affirmatively mended by \citet{heffetz2021}, who introduces an extra
treatment causing expectations to ``sink in," the evidence on real-effort experiments with EBLA
\citep{abeler2011,gneezy2017} does not allow a clear verdict, yet.
EBLA has been applied to a variety of economic models.\footnote{Such as moral hazard \citep{herweg2010},
monopoly pricing \citep{herweg2013,heidhues2014,carbajal2016}, pricing with competition
\citep{heidhues2008,karle2014}, consumer search \citep{karle2020}, and auctions \hbox{\citep{lange2010,rosato2019,wangenheim2021}}.} 

\textcolor{black}{The remainder of the paper is structured as follows. In Section \ref{sec:model}, we present the model environment, introduce reference-dependent preferences, and describe the appropriate equilibrium concept. In Section \ref{sec:CPE}, we apply this \textcolor{black}{decision-theoretic} equilibrium concept to the individual decision problem in strategy-proof mechanisms. We introduce the attainability distribution as a reduced form summarizing all information relevant to determine the optimal ROL. Next, we characterize all rationalizable ROLs. In Section \ref{sec:CBNE}, \textcolor{black}{we analyze strategic interaction in the static student-proposing deferred acceptance mechanism, and we show that a game-theoretic equilibrium exists. We then characterize equilibrium in a stylized setting with district and elite schools. In Section \ref{sec:remedy}, we establish that only a sequential mechanism can solve potential instability issues of the static DA mechanism.} All proofs are relegated to the appendix, Section \ref{sec:proofs}. }


\section{The Model} \label{sec:model}

\textbf{Players:}
We consider finite sets of students, $\mathcal I := \{i_1, \dots,i_n\}$,
and schools, $\mathcal S := \{1, \dots,m\}$.
Each school $s \in \mathcal S$ has a capacity of $q_s \in \mathbb N$ seats for students.
If we want to allow for students to remain unmatched, we can think of school $m$
as a safe outside option with unlimited capacity.

\textbf{Preferences:}
Each student $i \in \mathcal I$ privately draws a type $\theta_i = ( \mathbf v_i , \mathbf w_i, \Lambda_i)$, where each entry of vector $\mathbf v_i =  (v_{i,s})_{s \in \mathcal S}$ represents the payoff student $i$ receives from being matched with corresponding school $s$.\footnote{\textcolor{black}{\sout{In order to evaluate reference-dependent utility, we must rely on cardinal utilities. Yet, our main results will not depend on the cardinal ranking.}}}
Similarly, each element of vector $\mathbf w_i =  (w_{i,s})_{s \in \mathcal S}$ represents the payoff school $s$ receives from being matched with student $i$.
We explain the loss-aversion parameter $\Lambda_i\geq 1$ in its own section later. 
Let for all $i \in \mathcal I$ the type \textcolor{black}{$\theta_i$} be distributed according to a commonly observed distribution over a compact subset of \textcolor{black}{
$\mathbb R^m \times \mathbb R^{m} \times [1,\infty)$}.
\textcolor{black}{Unless explicitly stated otherwise, correlation of types between different students may be arbitrary.}

The ordinal preference over schools corresponding to type $\theta_i$
is captured by a rank-ordered list (ROL).
Formally, an ROL is a permutation of set $\mathcal S$, where ROL $(s_1,s_2, \dots, s_m)$
is interpreted as school $s_1$ being most preferred, $s_m$ least preferred, and $s_k$ having $k$-th highest preference.\footnote{Ties in the ROL may be arbitrarily broken. With continuous type distributions indifferences occur with probability zero and do not affect any result in this paper.}
Let $\mathfrak{S}(\mathcal{S})$ be the set of all such permutations.

\textbf{Mechanism:}
Our results refer to \textcolor{black}{a static mechanism, in which students only report ROLs, and this mechanism is strategy-proof for students with respect to standard preferences $(\mathbf{v_i})_{i \in \mathcal I}$.}
We \textcolor{black}{fix a mechanism, and} assume that schools always report their true \textcolor{black}{priorities} over students.\footnote{This assumption
distinguishes school choice where local laws determine schools' priorities from the college admission
problem where colleges are strategic actors, see, e.g., \citet{chen2006}.}
Formally, a \textcolor{black}{(pure)} reporting strategy for student $i$ is a mapping $\sigma_i: \Theta_i \rightarrow \mathfrak{S}(\mathcal{S})$ from types into ROLs.

\textbf{Properties:}
An allocation $M$ is a many-to-one mapping from $\mathcal I$ to $\mathcal S$
such that $M(i)=s$ denotes that student
$i$ is matched to school $s$  and $M^{-1}(s) = \{i: M(i)=s\}$ lists the students matched to $s$.
Feasibility requires $|M^{-1}(s)|\leq q_s$.
Let $\mathcal M$ be the set of all feasible allocations.
An allocation rule is a function $\alpha: \mathfrak{S}(\mathcal{S})^{n} \rightarrow \mathcal M$,
mapping profiles of ROLs  into matchings.
Let $\nu = (\nu_i)_{i \in \mathcal I}$ be the profile of true ROLs.
An allocation rule is strategy-proof if 
\eq{
v_{i,\alpha(\nu)[i]} \geq v_{i,\alpha(\nu'_i, \nu_{-i})[i]} \quad \forall \: i \in \mathcal I, \forall  \theta.
}
An allocation $M$ is stable\footnote{We use the classic definition of pairwise stability and, \textcolor{black}{in line with much of the applied matching literature,} use this word synonymously with ``no justified envy'' \textcolor{black}{with respect to standard preferences}, but we focus on the latter meaning, i.e., interpreting it as a fairness notion, following \citet{abdulkadirouglu2003} in the spirit of \citet{balinski1999}. To interpret stability as ``no  coalition  can  profitably  deviate" in our setting, we would have to take into account that  student $i$'s approach to school $s$ would create expectations and therefore scope for disappointment. \textcolor{black}{For a careful discussion of these two terms, see \citet{romm2020stability}.}}
if there is no pair $i,s$ such that student $i$ has justified envy,
\eq{
v_{i,s}  >   v_{i,M(i)} \quad \mbox{and} \quad
w_{s,i} >  w_{s,i'} \quad \mbox{for some } i' \in M^{-1}(s),
}
i.e., no student $i$ prefers another school $s$ over her match, while this school
prefers $i$ over at least one of her matched students.
A student-optimal stable matching is a stable matching $M$ such that
\eq{
v_{i,M(i)} \geq v_{i,M'(i)} \quad \mbox{ for any stable matching } M'.
}

\textbf{Loss aversion:}
Each student reports the preferences maximizing her expected utility. Students are expectation-based loss averse in the sense of \cite{koszegi2006,koszegi2007}.  In addition to classical match utility $v_{i,s}$, the student \textcolor{black}{experiences psychological utility in form of perceived} gains and losses when comparing the realized match utility to her reference utility.
For the specification of gain-loss utility, we follow most of the literature by assuming a linear gain-loss function with a kink at zero.
More specifically, let
\eq{ \label{eq:utility}
u ( \theta_i , s | r) = v_{i,s} + \begin{cases}
\eta_i  (v_{i,s} - v_{i,r}) \quad &\mbox{ if } v_{i,s} \geq v_{i,r},\\
 \eta_i \lambda_i (v_{i,s} - v_{i,r}) \quad &\mbox{ if } v_{i,s} < v_{i,r},
\end{cases}
}
denote student $i$'s ex-post utility from being matched with school $s$, when school $r\in \mathcal S$
is her reference match.
The parameter $\lambda_i>1$ captures the degree of loss aversion, whereas $\eta_i\geq 0$
is the weight assigned to the gain-loss
utility. As we will show in \eqref{eq:payoff-distance}, behavior is driven by a summarizing parameter 
$\Lambda_i =\lambda_i \eta_i-\eta_i$ called the loss dominance.
We call students with $\Lambda_i \leq 1$ moderately loss averse and students with $\Lambda_i>1$ dominantly
loss averse.

\textcolor{black}{In Section \ref{sec:CPE}, we consider the individual decision problem of a student $i$ with any fixed $\theta_i$. We define an attainability distribution and delineate how, as a reduced form, this distribution sums up all relevant information about beliefs with respect to priorities and other students' submitted ROLs. Given this attainability distribution,} each report $\sigma_i (\theta_i)$ corresponds to a distribution $F_i = (f_{i,s})_{s \in \mathcal S}$, where $f_{i,s}$ denotes the probability with which $i$ expects to be matched with school $s$. We say a lottery is feasible for student $i$ if there exists an ROL that induces it, 
and let $\mathcal F_i (\theta_i,\sigma_{-i})$ be the set of feasible lotteries.

\textcolor{black}{Following \cite{koszegi2006,koszegi2007}, agents hold stochastic reference points, formed by beliefs over outcomes. When comparing an outcome to this reference distribution, they assign gains and losses to each comparison from the reference distribution. For instance, if a student expects a match utility of 1, 2 or 3 with equal probability, a realized match utility of 2 would produce a feeling of gain of one util weighted with probability $1/3$, and a feeling of loss of one util weighted with probability $1/3$. Taking the expectation over these utilities for all possible realizations,} the expected utility from a lottery $F_i$
evaluated with respect to some reference lottery $G=(g_{s})_{s \in \mathcal S}$ is therefore
\begin{equation}\label{EU}
\mathcal U_i  (\theta_i, F_i| G)  =
\sum\limits_{s \in \mathcal S} f_{i,s} \left( \sum\limits_{r \in \mathcal S}
u (\theta_i, s|r)  g_r \right).
\end{equation}


\textbf{Equilibrium:}
\textcolor{black}{An ROL}
 is a choice-acclimating personal equilibrium (CPE) for student $i$ \textcolor{black}{given her type $\theta_i$ and an attainability distribution} if the corresponding distribution $F_i$ satisfies,
\eq{ \label{eq:interim-IC}
\begin{aligned}
U_i(\theta_i, F_i) :=
\mathcal U_i  (\theta_i, F_i | F_i )
\geq
\mathcal U_i  (\theta_i, F'_i | F'_i) 
\quad \: \forall F'_i \in \mathcal F_i(\theta_i,\sigma_{-i}).
\end{aligned}
}
That is, we assume expectation-based loss aversion (EBLA) according to \citet[Section IV]{koszegi2007}, where
the reference point is stochastic and determined by the actual belief \textcolor{black}{about} the own match outcome.
In a CPE, \textcolor{black}{the ROL} maximizes expected utility
given that the corresponding beliefs determine both the reference lottery and the outcome lottery. 

\textcolor{black}{Our choice of equilibrium concept comes with an assumption on timing. After learning the mechanism's rules, students draw their own types, and form beliefs about others' types and the priorities at schools. By choosing an ROL next, the student implicitly commits to the outcome lottery corresponding to this ROL and the distribution of ROLs she expects from other students. Consequently, this lottery is the reference point according to which gains and losses are defined and weighted, and submitting another ROL would amount to having a different reference point.}

\textcolor{black}{\cite{koszegi2006,koszegi2007} also suggest an alternative equilibrium concept called unacclimating personal equilibrium (UPE), where the reference distribution does not adjust when contemplating deviations. Both concepts are frequently used in the literature. We opt for the choice-acclimated version for various reasons. First, it captures most closely the notion of disappointment, as discussed in the introduction and below \eqref{eq:payoff-distance}. Second, only CPE can explain stochastically dominated choices such as non-truthful play in a truthful mechanism. Third, CPE reflects that by submitting an ROL, a student \textit{commits} to a lottery over outcomes with respect to which disappointment is evaluated, whereas UPE is appropriate in situations where choices are anticipated but not committed to in advance.}

\textcolor{black}{In Section \ref{sec:CBNE}, we consider strategic interaction and define the appropriate game-theoretic equilibrium concept, 
choice-acclimating Bayesian Nash equilibrium (CBNE). That is, while CPE is a decision-theoretic concept requiring that a student best-responds to given beliefs, CBNE additionally requires that each student's beliefs are consistent with all the other students' strategies.}


\section{The Individual Decision Problem} \label{sec:CPE}

As we consider the individual problem of some representative student $i$ by fixing her type $\theta_i$ and the
other students' strategy profile $\sigma_{-i}$, it is convenient to drop the student's indices $i$ and also,
without loss of generality, relabel schools such that $v_1 > v_2 > \dots > v_m$.

\subsubsection*{Match probabilities and attainability} \label{sec:attainability}

\textcolor{black}{
For a fixed mechanism and given the other students' ROLs, schools' priorities and capacities, we call school $s$ attainable for student $i$ if there exists some ROL such that the given mechanism matches $i$ and $s$. In a strategy-proof mechanism, this is the case if and only if $i$ is assigned to $s$ when ranking it first. Let $A_{s} \in \{1,0\}$ be a binary variable determining whether school $s$ is attainable for the representative student. The attainability distribution $P$ is a probability distribution over attainability states $(A_s)_{s \in \mathcal S}$.}

\textcolor{black}{Example \ref{ex} illustrates the concept of the attainability distribution for the static student-proposing DA mechanism.} By construction of DA, student $i$ is rejected by school $s$ if at some step of the algorithm more than $q_s$ students with a higher priority than $i$ apply to school $s$. Hence, student $i$ is matched to the $k$-th ranked school of her ROL if the capacities of all schools she ranked before are filled by students that these schools individually prefer over student $i$.
\textcolor{black}{By strategy-proofness of DA, there exists no ROL such that student $i$ is assigned to school $s$ for the given priorities and ROLs submitted by other students. Hence, the rejecting school $s$ is not attainable.}

\vspace{-8pt}\begin{lem} \label{lem:acc}
\textcolor{black}{A static strategy-proof mechanism} assigns a student to her highest-ranked attainable school.
\end{lem} \vspace{-8pt}

Whether a school is attainable for student $i$ depends on the strategies of other students and on the schools' \textcolor{black}{capacities and priorities } over students, but not on the ROL submitted by the student herself. The submitted ROL does, however, determine which of the attainable schools is ranked first, and hence constitutes the student's match. Therefore, the submitted ROL determines the match outcome distribution $F$, and selecting an ROL effectively corresponds to choosing a lottery over match outcomes. {\color{black}
From the perspective of student $i$, all the relevant information about the distribution of other students' ROLs and priorities needed to evaluate the expected payoffs \eqref{eq:interim-IC} is collected in attainability distribution $P$. Indeed, by Lemma \ref{lem:acc} match probabilities induced by an ROL $(s_1,s_2,\dots,s_m)$ are given by
\eq{\label{eq:probs}
f_{i,s_k}= \mbox{Pr} (A_{s_\ell}=0 \: \forall \ell<k, A_{s_k}=1).
} 
Hence, an individual's decision problem can be summarized by a set of schools $\mathcal S$, a preference vector $\mathbf{v_i}$, an individual loss parameter $\Lambda_i$, and an attainability distribution $P$.}

\textcolor{black}{In the following, w}e denote for each school $s$ with 
$p_{s}=P(A_{s}=1)$ the unconditional probability of being attainable at that school. Importantly, the attainability probabilities are usually not independent, even when types are independent draws.

\subsubsection*{Outside options and truncated lists} \label{sec:outside}

In many existing implementations of \textcolor{black}{strategy-proof mechanisms}, it is allowed to submit incomplete ROLs.
We can include the possibility to drop a school from the ROL by defining 
the outside option as a fictional school $m$ with
unlimited capacity and normalized $v_m=0$. Depending on the context the outside option may refer to remaining unmatched or being matched to some ``default'' option.
Different ROLs which rank the same schools in arbitrary order behind the outside option are equivalent in the sense that 
they induce the same match probabilities. 
Because the outside option is always attainable, ranking a school after it 
corresponds to dropping this school such that a match with this school is excluded.
While it is never optimal to list schools with $v_s<0$, 
we will show that it can be optimal to drop schools with $v_s>0$.
%

\subsubsection*{Payoffs}

For any ROL resulting in lottery $F=(f_1,...,f_m)$, we can rewrite  the expected utility 
as
\begin{align}
U ( \theta ,  F) & = \sum\limits_{s \in \mathcal S} f_s \left( \sum\limits_{r \in \mathcal S}
u (\theta, s|r)  f_r \right) \notag \\
 &= \sum_{s=1}^{m} f_s 
\left[
 \left( \sum_{r=1}^m f_r\right) v_s +  \sum_{r=1}^{s-1} f_r \lambda \eta (v_s-v_r)+
  \sum_{r=s+1}^m f_r \eta (v_s-v_r)  \right] \notag \\
   & = \underbrace{\sum_{s=1}^{m} f_s  v_s}_{\text{classical utility}} -\underbrace{\Lambda \sum_{s=1}^m\sum_{r =s+1}^m f_s f_r (v_s-v_r)}_{\text{gain-loss utility}}. \label{eq:payoff-distance}
\end{align}
Each pairwise comparison is weighted by $f_s f_r$ and shows up twice: once as a gain and once as a loss, 
its total factor is $\Lambda>0$.
Since losses are weighted stronger than gains, expected gain-loss
utility always enters negatively. 
Under our notational convention, the difference $(v_s-v_r)$ is positive for each $r>s$.
One can think of the expected gain-loss term as the
cost of uncertainty. It is proportional to the loss dominance $\Lambda$ and the
average distance between two realizations.
An equal weight on gains and losses, $\lambda = 1$, would result in $\Lambda =0$ such that
students only maximize classical utility.
If $\Lambda >1$, gain-loss utility may dominate match utility, which will become central.

\subsubsection*{Example} \label{sec:examples}

The following example illustrates the tradeoff between the gains from classical utility and the losses from expected reference-dependent utility, which provides the incentives to misrepresent true preferences. It foreshadows our characterization results
on which ROLs can be rationalized under EBLA and provides intuition for comparative statics in a student's loss dominance parameter and her priority.
Intuitively, increasing $\Lambda$ augments the relative weight of gain-loss utility over match utility.
Hence, reducing the exposure to sensations of loss by taming expectations becomes a central motif.

\vspace{-8pt}
\begin{example} \label{ex}
There are three students, \textcolor{black}{$\mathcal I = \{X,Y,Z\}$}, and two schools with a single seat each, such
that one student will remain unmatched. By treating the outside option as a third school with unconstrained capacity, we obtain $\mathcal S = \{1,2,3\}$
with capacities $q_1=q_2=1,q_3=3$.
Suppose that all students prefer a school seat over being unmatched and that school 1
is expected to be the more popular school,
\ea{
\mbox{Pr} ( v_{i,1} > v_{i,2} > v_{i,3} ) = (1- \varepsilon) \quad > \quad
\mbox{Pr} ( v_{i,2} > v_{i,1} > v_{i,3} ) = \varepsilon \quad \forall i \in \mathcal I.
}
\textcolor{black}{Priorities} are determined by scores $\omega_i$ which each student independently
draws from a uniform distribution on $[0,1]$.
We take the perspective of student \textcolor{black}{$X$} with preferences $v_1 > v_2> v_3$ and score $\omega$.
Suppose she believes the other two students are truthful.
Table \ref{tab:attainable} provides the attainability \textcolor{black}{distribution if} 
$\omega=\nicefrac{1}{4}$ and $\varepsilon=\nicefrac{1}{20}$. \textcolor{black}{The fictional school 3 is always attainable, $A_3=1$.}

\begin{table}[h]
\begin{center}
 \begin{tabular}{||c c c||}
 \hline
 Attainability & ~~~~~\textcolor{black}{$A_1=1$} & \textcolor{black}{$A_1=0$} \\ [1ex]
 \hline\hline
\textcolor{black}{$A_2=1$} & ~~~~~~~~~~$\omega^2=\nicefrac{10}{160}$ & ~~$2\omega(1-\omega)(1-\varepsilon)=\nicefrac{57}{160}$~~~ \\
 \hline
 \textcolor{black}{$A_2=0$} & $2\omega(1-\omega)\varepsilon=\nicefrac{3}{160}$ & ~~~~~~~~~~~~$(1-\omega)^2=\nicefrac{90}{160}$~~~ \\
 \hline
\end{tabular}
\caption{Attainability probabilities for $\omega=\nicefrac{1}{4}$ and $\epsilon=\nicefrac{1}{20}$
and $\omega=\nicefrac{1}{4}$.} \label{tab:attainable}
\end{center}
\end{table}

 Evidently, both schools are attainable for \textcolor{black}{$X$} only if she has the highest score, and neither school is if she has the lowest score. Only one of the schools is attainable if she has the second highest score and the student the with highest score prefers the other school. Note that the attainability probabilities are interdependent, even though preferences and scores are drawn independently.

From the attainability \textcolor{black}{distribution}, the student can infer the lottery over match outcomes for any possible ROL. For instance, the true ROL, $(1,2,3)$, leads to a  match with school 1 if and only if it is attainable, with school 2 if and only if is attainable but school 1 is not, and to no match if and only if both schools are unattainable. Table \ref{tab:ex} presents match probabilities for all ROLs.

\begin{table}[h]
\begin{center}
 \begin{tabular}{||c c c c||}
 \hline
 ~~ROL~~ & ~~$f_1$~~ & ~~$f_2$~~ & ~~$f_3$~~\\ [1ex]
 \hline\hline
 1,2,3 & $\nicefrac{13}{160}$ & $\nicefrac{57}{160}$ & $\nicefrac{90}{160}$\\
 \hline
 2,1,3 & $\nicefrac{3}{160}$ & $\nicefrac{67}{160}$ & $\nicefrac{90}{160}$\\
 \hline
 2,3,1 & $0$ & $\nicefrac{67}{160}$ & $\nicefrac{93}{160}$\\
 \hline
 1,3,2 & $\nicefrac{13}{160}$ & $0$ & $\nicefrac{147}{160}$\\
 \hline
 3,2,1 & $0$ & $0$ & $1$\\
  \hline
 3,1,2 & $0$ & $0$ & $1$\\
 \hline
\end{tabular}
\caption{All possible ROLs of the example and the corresponding lotteries for $\epsilon=\nicefrac{1}{20}$
and $\omega=\nicefrac{1}{4}$.} \label{tab:ex}
\end{center}
\end{table}

We see that flipping 1 and 2 in the ranking
shifts a probability mass of $\nicefrac{10}{160}$ (the probability of both schools being attainable) from school 1 to 2, which decreases not only classical utility but also the cost of uncertainty.
Similarly, dropping the last ranked school simply shifts match probability mass from this school to the outside option. 
ROLs listing the outside option first induce identical degenerate lotteries.

\begin{figure}[h!]
\begin{center}
\begin{adjustbox}{center,scale=0.75}
\subfloat[Expected utility as a function of $\Lambda$ with $\omega=\nicefrac{1}{4},$ ]{ \label{subfig:lambda}
\includegraphics{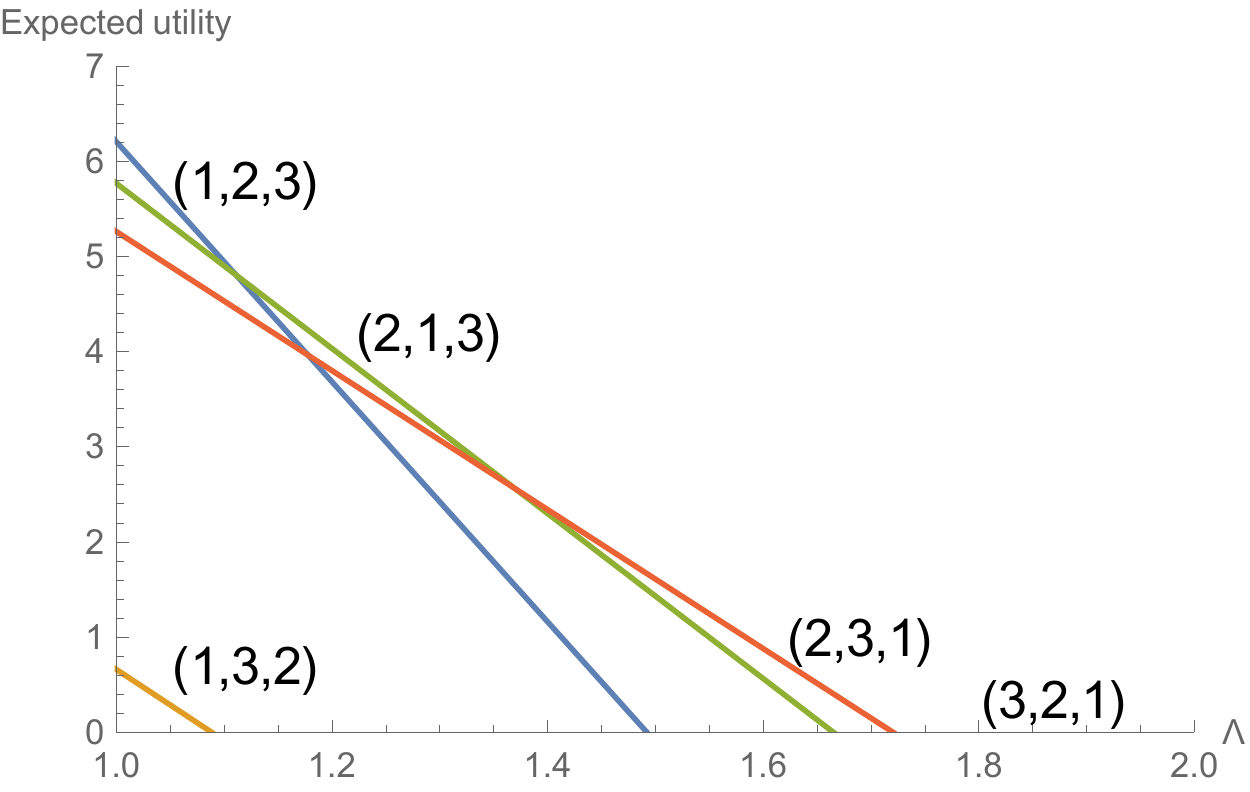}
}
\subfloat[Expected utility as a function of $\omega$ with $\Lambda=\nicefrac{3}{2}$.]{ \label{subfig:delta}
\includegraphics{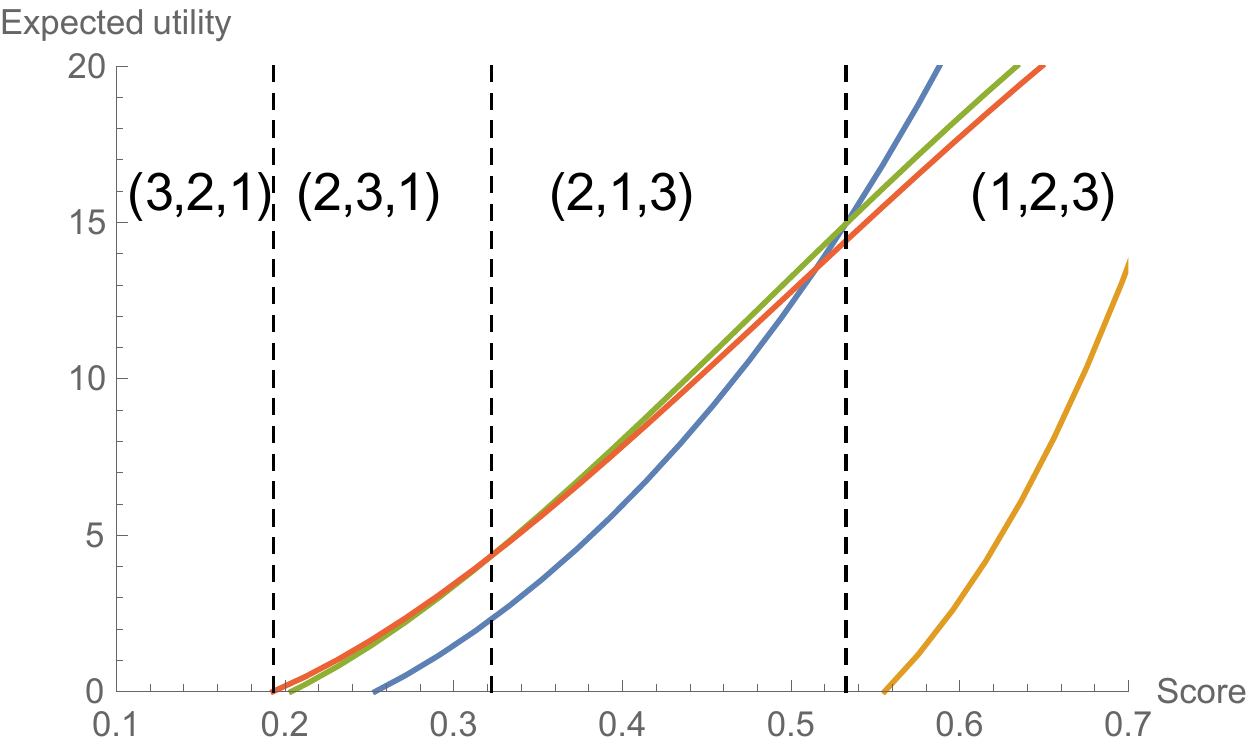}
}
\end{adjustbox}
  \caption{The expected utilities induced by every ROL as a function of (a) $\Lambda$
  and (b) $\omega$, setting $v_1=100, v_2=30, v_3=0$ and $\varepsilon=\nicefrac{1}{20}$.}\label{fig:ex}
\end{center}
\end{figure}

Given the lotteries, we can calculate the expected utilities for any $\Lambda$ and select the optimal ROL.
Figure \ref{fig:ex} illustrates the expected utilities induced by different ROLs.
Figure \ref{subfig:lambda} demonstrates that for a sufficiently small $\Lambda$ the student always reports truthfully, as the lottery corresponding to the true ROL first-order stochastically dominates every other lottery and the positive effects on match utility dominate the cost of uncertainty. As we increase $\Lambda$, 
preferred schools are optimally ranked as worse, ultimately culminating in submitting an empty ROL when the perceived cost of uncertainty is sufficiently high.\footnote{Abstaining from the mechanism by choosing a dominated outside option is reminiscent of the ``uncertainty effect" documented by \citet{gneezy2006}.} Notably, any optimal manipulation involves a flipping (or dropping) of the most preferred option -- ROL $(1,3,2)$ is never optimal. 
 Figure \ref{subfig:delta} illustrates that students tend to become more truthful as
their scores increase and they become more optimistic. 

A large $\Lambda$ by itself does not lead to profitable deviations from the true ROL, as they are inherently linked to incomplete information. 
If students had full information about other students' reports and \textcolor{black}{priorities},
they could infer their match outcome from the mechanism, and would have no cost of being truthful, such that DA would implement the student-optimal stable matching.
Moreover, the optimal ROLs of this example cannot be explained by simple risk aversion,
because the truthful lottery first-order stochastically dominates every other feasible lottery.

\end{example}

\subsubsection*{Characterization of optimal ROLs} \label{sec:optimalROL}

As we have seen in Example \ref{ex},
the dominance of the truthful strategy does not necessarily carry over to a truthful CPE if loss aversion is sufficiently strong. As we will show in Proposition \ref{prop:bounds}, for any $\Lambda>1$, a sufficiently pessimistic student will misrepresent her preferences. 
Conversely, \citet[Proposition 1]{masatlioglu2016} show that CPE respects first-order stochastic dominance if 
$\Lambda \leq 1$, and, by the dominance of the truthful strategy with standard preferences, 
the truthful lottery first-order stochastically dominates any other feasible lottery.
Hence, the truthful strategy is a CPE for any profile and all possible beliefs if and only if 
$\Lambda \leq 1$.

There is substantial evidence that a relevant fraction of the population is indeed dominantly loss averse.\footnote{While many applied papers restrict attention to $\Lambda \leq 1$, we explicitly allow (all or only some) students to be dominantly loss averse, in order to explain deviations from truth-telling. \citet{herweg2010} introduced the assumption ``no dominance of gain-loss utility'' as $\lambda \leq 2$ with fixed
	$\eta=1$, and it was later picked up in various forms. Rather than being based on evidence, the main reason why it is imposed seems to be that it makes problems well-behaved.}
$\Lambda>1$ matches the conventional wisdom that ``losses loom about twice as large as gains." While this rule of thumb
originates from studies on riskless choices, it also seems to apply when risk is involved.\footnote{See \citet{tversky1992}, \citet{gill2012}, \citet{sprenger2015} or \citet{karle2015}.}
In a meta-analysis of over 150 articles, \citet{brown2021} find that the mean loss-aversion coefficient $\lambda$ with $\eta=1$ is between 1.8 and 2.1 and about 38\% out of more than 600 estimates find $\lambda>2$, corresponding to $\Lambda>1$.
While the possible preference for first-order stochastically dominated lotteries that comes with
this assumption may appear counterintuitive, it is not only observable in the matching context.\footnote{See
the discussion around Proposition 7 by \citet{koszegi2007}.
While the ``uncertainty effect" found by \citet{gneezy2006} provides evidence in this direction,
\hbox{\citet{rydval2009}} suggest it cannot be replicated.
In the context of choice bracketing, \hbox{\citet{tversky1981}} and \citet{rabin2009}
provide experimental evidence that people can have a preference for dominated lotteries.}

When searching for utility maximizing ROLs the following property turns out to be both, a necessary and sufficient condition for optimality \textcolor{black}{for some cardinal utility vector}.

\begin{definition}
\textcolor{black}{A ROL $(s_1,s_2,\dots,s_m)$ is top-rank monotone if it has the following properties:
\begin{itemize}
\item For any school preferred over the top-ranked school, $s_k<s_1$, it must be that all schools preferred over $s_k$ are ranked behind $s_k$, $s_i<s_k \Rightarrow i>k$.
\item For any school not preferred over the top-ranked school, $s_\ell>s_1$, it must be that all schools preferred over $s_\ell$ are ranked before $s_\ell$, $s_j<s_\ell \Rightarrow j<\ell$.
\end{itemize}}
\end{definition}

\textcolor{black}{That is, top-rank monotone ROLs reverse the order of schools preferred over the top-ranked school, because the rank of these schools is decreasing in their rank in the true ROL. Conversely, they preserve the order of schools not preferred over the top-ranked school, because the rank of these schools is increasing in their true rank.} For example, the ROL $(1,3,2,4)$ is not top-\textcolor{black}{rank} monotone, because 2 is ranked behind 3 although $v_2>v_3$ such that the preference order of schools considered worse than \textcolor{black}{the top-ranked school} 1 is not reflected in the ranking. Similarly, ROL $(3,1,2,4)$ violates the property, while ROL $(3,2,1,4)$ satisfies it as the preference ranking of schools preferred over \textcolor{black}{the top-ranked school} 3 is reversed. Table \ref{tab:ROL} exhibits further examples.

\textcolor{black}{The following condition implies that all ROLs correspond to different lotteries such that swapping two adjacently ranked schools' positions in the ranking alters their match probabilities \eqref{eq:probs} by shifting a probability mass $\varepsilon>0$.}
\textcolor{black}{\begin{definition} \label{def:full-sup}
An attainability distribution $P$ has full support if for all $(a_1,...,a_m)\in\{0,1\}^m$ the distribution satisfies $\text{Pr}(\forall i\in\{1,...,m\}:A_i=a_i)>0$.
\end{definition}}

\begin{prop} \label{prop:order}
Take any $\mathcal S=\{1,\dots,m\}$ with the implied ordinal ranking and any $\Lambda>1$.
\begin{itemize}
    \item[a)] For any \textcolor{black}{attainability distribution $P$, any strictly optimal ROL must be top-rank monotone. For any attainability distribution $P$ with full support, any optimal ROL must be top-rank monotone.}
    
    \item[b)] For any ROL $L$ which is top-\textcolor{black}{rank} monotone with respect to the ordinal ranking,  there is 
    \textcolor{black}{an attainability distribution $P$ and a cardinal utility vector with $v_1>v_2>...>v_m$} such that $L$ is strictly optimal.
    \end{itemize}
\end{prop}

If some ROLs correspond to identical lotteries (and therefore identical expected utility), it is possible
that a student is indifferent between multiple ROLs out of which at least one is top-\textcolor{black}{rank} monotone.\footnote{For this reason, we render ROLs equivalent for which only the ranking after the outside option differs. In Table \ref{tab:ROL} the darkly shaded ROLs are in this sense redundant as they represent the same lottery as a unique top-\textcolor{black}{rank} monotone analog.
Identical lotteries can also arise if a subset of schools together constitute an outside option, making any permutation of schools
ranked after them meaningless. Similarly, the ranking of two schools that are not attainable does not matter.
There are no equivalent ROLs if
for any subset of schools the probability of all of them being attainable is strictly between zero and one.}
Only a comparably small set of ROLs is top-\textcolor{black}{rank} monotone.\footnote{
Indeed, while for $m$ schools the number of
ROLs is $m!$ (or $\sum_{i=1}^{m} (m-i)! { m-1 \choose i-1 } = \sum_{i=1}^{m} \nicefrac{(m-1)!}{(i-1)!}$
non-redundant ROLs when $m$ is an outside option),
just $2^{m-1}$ are top-\textcolor{black}{rank} monotone.}

The formal proof of the proposition is the appendix, but
its general idea is easily understood by example.
Table \ref{tab:ROL} shows all possible ROLs for a setting with $m=4$ as an outside option.
The bold numbers are the listed schools,  as schools ranked after 4 can be interpreted as ``dropped from the ranking.'' 
The shaded ROLs are 
never strictly optimal as they violate top-\textcolor{black}{rank} monotonicity. 
For instance, $(1,3,2,4)$ reverses 2 and 3 which are not preferred over \textcolor{black}{the top-ranked school} 1. 
Intuitively, if the student were willing to reduce risk by shifting probability mass from school 2 to 3, i.e., $(1,3,2,4) \succ_i (1,2,3,4)$, then she would be a forteriori willing to shift probability mass from school 1 to 3, i.e., 
$(3,1,2,4) \succ_i (1,3,2,4)$. Hence, $(1,3,2,4)$ can never be strictly optimal.
Similarly, $(3,1,2,4)$ cannot be strictly optimal as either
$(1,3,2,4) \succ_i (3,1,2,4)$ or $(3,2,1,4) \succ_i (3,1,2,4)$.

\begin{table}[h]
\begin{center}
\begin{tabular}{|c|c|cl}
\hline
\textbf{Full ROL} & \textbf{Drop one} & \multicolumn{1}{c|}{\textbf{Drop two}} & \multicolumn{1}{c|}{\textbf{Empty ROL}}    \\ \hline
\textbf{\textbf{1,2,3}},4    & \cellcolor{lightgray}\textbf{\textbf{1,2}},4,3      & \multicolumn{1}{c|}{\cellcolor{lightgray}\textbf{\textbf{1}},4,3,2}        & \multicolumn{1}{c|}{4,3,2,1} \\ \hline
\textbf{\textbf{2,1,3}},4    & \cellcolor{lightgray}\textbf{\textbf{2,1}},4,3      & \multicolumn{1}{c|}{ \cellcolor{lightgray}\textbf{\textbf{2}},4,3,1}        &                                           \multicolumn{1}{c|}{\cellcolor{gray}4,1,2,3}   \\ \hline
\cellcolor{lightgray}\textbf{\textbf{3,1,2}},4    & \cellcolor{lightgray}\textbf{\textbf{3,1}},4,2      & \multicolumn{1}{c|}{\textbf{\textbf{3}},4,2,1}        &                                           \multicolumn{1}{c|}{\cellcolor{gray}4,2,1,3}  \\ \hline
\cellcolor{lightgray}\textbf{\textbf{1,3,2}},4    & \cellcolor{lightgray}\textbf{\textbf{1,3}},4,2      & \multicolumn{1}{c|}{\cellcolor{gray}1,4,2,3}                   &                                           \multicolumn{1}{c|}{\cellcolor{gray}4,3,1,2}   \\ \hline
\textbf{\textbf{2,3,1}},4    & \textbf{\textbf{2,3}},4,1      & \multicolumn{1}{c|}{\cellcolor{gray}2,4,1,3}                   &                                           \multicolumn{1}{c|}{\cellcolor{gray}4,1,3,2}  \\ \hline
\textbf{\textbf{3,2,1}},4    & \textbf{\textbf{3,2}},4,1      & \multicolumn{1}{c|}{\cellcolor{gray}3,4,1,2}                   &                                           \multicolumn{1}{c|}{\cellcolor{gray}4,3,2,1}   \\ \hline
\end{tabular}
\caption{All possible permutations with three schools and an outside option.
The darkly shaded ROLs are redundant.
The lightly (and darkly) shaded ROLs are not top-\textcolor{black}{rank} monotone and thus never strictly optimal.} \label{tab:ROL}
\end{center}
\end{table}

As a first impression of our theory's predictive power, we briefly consider an experiment by
\citet[treatment SP-RSD]{li2017}.
Here, each participant is privately endowed with a priority score (an integer between 1 and 10)
and is informed about how all participants commonly value each of four prizes between \$0 and \$1.25.
Then, participants simultaneously submit an ROL about the prizes to a mechanism 
which calculates the DA allocation.

\begin{table}[h!]
\begin{adjustbox}{center,scale=1}
\begin{tabular}{|l|lllllllllll|}
\hline
Priority & \multicolumn{1}{c}{1}                & \multicolumn{1}{c}{2} & \multicolumn{1}{c}{3} & \multicolumn{1}{c}{4} & \multicolumn{1}{c}{5} & \multicolumn{1}{c}{6} & \multicolumn{1}{c}{7} & \multicolumn{1}{c}{8} & \multicolumn{1}{c}{9} & \multicolumn{1}{c}{10} & \multicolumn{1}{c|}{ALL} \\ \hline
~~~~1234     & 61.1\%                               & 57.1\%                & 58.8\%                & 67.7\%                & 55.2\%                & 79.0\%                & 74.4\%                & 85.7\%                & 84.3\%                & 91.3\%                 & 71.0\%                   \\
\multicolumn{1}{|r|}{2134}     & 1.1\%                                & 1.2\%                 & 3.8\%                 & \textbf{6.5\%}        & \textbf{12.1\%}       & \textbf{8.1\%}        & \textbf{10.3\%}       & \textbf{7.1\%}        & \textbf{5.7\%}        & \textbf{1.3\%}                  & \textbf{5.3\%}           \\
\multicolumn{1}{|r|}{3214}     & 6.7\%                                & 6.0\%                 & \textbf{7.5\%}        & 4.8\%                 & 3.4\%                 & 0.0\%                 & 0.0\%                 & 1.8\%                 & 0.0\%                 & 0.0\%                  & 3.2\%                    \\
\multicolumn{1}{|r|}{4321}     &\textbf{17.8\%} & \textbf{8.3\%}        & 3.8\%                 & 4.8\%                 & 1.7\%                 & 3.2\%                 & 1.3\%                 & 0.0\%                 & 2.9\%                 & 0.0\%                  & 4.9\%                    \\ \cline{1-12}
T\textcolor{black}{R}M      & 91.1\%                               & 77.4\%                & 77.5\%                & 88.7\%                & 75.9\%                & 91.9\%                & 87.2\%                & 98.2\%                & 95.7\%                & 93.8\%                 & 87.5\%                   \\ \hline
\end{tabular}
\end{adjustbox}
\caption{Share of most commonly submitted ROLs and total share of top-\textcolor{black}{rank} monotone ROLs for each priority score. Most common deviation from truth-telling for each priority score, i.e., for each column in Table \ref{tab:li2} in the appendix, is in bold.}
\label{tab:li1}
\end{table}

Table \ref{tab:li1} summarizes several noteworthy observations regarding our theoretical results.
Table \ref{tab:li2} in the appendix provides more details.
\textcolor{black}{While 71\% of the submitted ROLs are truthful (first row, last column), 87.5 \% of ROLs are top-rank monotone (last row, last column), i.e., can be strictly optimal for some attainability distributions.}
More importantly, the most common misrepresentations for each priority score (for each priority score in bold face)
are indeed all top-\textcolor{black}{rank} monotone.
Moreover, the rates of these misrepresentations move according to the intuitions suggested by our model.
ROL $(4,3,2,1)$ is most common among low scores, ROL $(3,2,1,4)$ among lower intermediate scores,
and ROL $(2,1,3,4)$ among higher intermediate scores.
As suggested by Example \ref{ex}, higher scores are more likely to submit truthful ROLs. 

Proposition \ref{prop:order} implies that any manipulation of the ROL 
will concern the most preferred schools:
the true ROL is strictly optimal if and only if it is strictly optimal to rank school 1 first.
This insight helps us to provide necessary and sufficient conditions on the loss parameter
which determine whether a manipulation of the true ROL is profitable.
The attainability probability $p_1$ only depends on beliefs about what other students do and
their priority relative to our representative student.
Hence, Proposition \ref{prop:bounds} gives precise bounds on when \textcolor{black}{mechanisms that are strategy-proof with respect to standard preferences are} incentive-compatible
for loss-averse students based only on fundamentals that are exogenous in this section.
These bounds are strict in the sense that for any $p_1\in[\frac{1-\nicefrac{1}{\Lambda}}{2},1-\nicefrac{1}{\Lambda}]$ the answer to whether truthfulness is optimal depends on other attainability probabilities and also the cardinal utilities.

\begin{definition}
\textcolor{black}{We say a school $i$ is exclusive if $A_i=1$ implies $A_j=0$ for some $j\neq i$.}\footnote{\textcolor{black}{Hence, if a school is exclusive then attainability at that school implies non-attainability at some other school.} For instance, a boy school would be exclusive, if there were a girl school in the set of schools. Evidently, if a school is  exclusive with other schools, then the rank of the  exclusive school among these schools in an ROL is inconsequential for attainability, and hence multiple ROLs induce the same outcome lottery. \textcolor{black}{While this condition is implied by the condition in Definition \ref{def:full-sup}, it is much weaker.}}
\end{definition}


\begin{prop}\label{prop:bounds}
Suppose the most preferred school is not exclusive. Let $p_1$ be the probability that a student's most preferred school is attainable.
\begin{enumerate}
  \item If $p_1>1-\nicefrac{1}{\Lambda}$, then the  true ROL is strictly optimal.
  \item If $p_1<\frac{1-\nicefrac{1}{\Lambda}}{2}$, then the true ROL is strictly suboptimal.
\end{enumerate}
\end{prop}

\textcolor{black}{This proposition implies that students with sufficiently high priority at their preferred school} report truthfully 
whereas \textcolor{black}{ students with sufficiently low priority} misrepresent whenever seats at their preferred school are scarce.
This result is in line with the evidence that suggests a causal relationship between priority
and truthfulness mentioned in our introduction and Table \ref{tab:li1}.

An important implication of the result is that students' beliefs are crucial.
That is, one of the advantages of strategy-proof mechanisms, namely, the irrelevance
of priors, vanishes in our setting.
Importantly, we have made no assumptions on whether the beliefs determining the attainability
probabilities are rational.
Consequently, EBLA is a channel which renders other well-documented biases
distorting the beliefs as decisive.
Here, an overconfident student is more likely to be truthful as she overestimates
her chances of getting into her favorite school.
Hence, overconfidence and loss aversion countervail each other in terms of incentive compatibility.
Indeed, \citet{rees2018b} find that overconfident participants are more likely to be truthful.\footnote{In their online experiment, participants
completed a test on logical reasoning ability and afterwards estimated the percentage of other
participants they outperformed. They deem a participant overconfident if they
overestimated their percentile rank.}
Without our theory, this observation may appear counterintuitive as this bias usually steers behavior
away from a rational unbiased benchmark.

\textcolor{black}{Incomplete ROLs} are prevalent \textcolor{black}{even when truncation is not mandatory}.
Since constraining the ROLs to a fixed number of schools can destroy strategy-proofness, economists often advocate against such restrictions.
While prohibiting complete ROLs introduces strategic motifs into \textcolor{black}{strategy-proof mechanisms} with standard preferences,
such motifs are already present in our setting. As in our models truncations formally  correspond to ranking a school behind the outside option, Proposition \ref{prop:order} implies the following statement on truncations:

\begin{cor} \label{cor:drop}
It is never strictly optimal to drop some desirable school $k$ from the ROL, 
but list some preferred school $\ell<k$.
\end{cor}

\section{Strategic Interaction \textcolor{black}{in DA}} \label{sec:CBNE}

{\color{black}
In this section, we fist show that choice-acclimating Bayesian Nash equilibria exist in our framework when type spaces are finite. We then turn to the static student-proposing deferred acceptance mechanism (DA) and rationalize the prevalent district-school bias as an equilibrium phenomenon in a stylized setting with district and elite schools.

\cite{dato2017expectation} show the general nonexistence of (mixed-strategy) choice-acclimating Bayesian Nash equilibria (CBNE). However, they make a particular assumption about the interpretation of mixed strategies. In their model, the player herself is uncertain about which (pure strategy) action from her mixed strategy she will use when her reference point forms. In the following, our interpretation of a mixed strategy rather follows, e.g., \cite{rubinstein1991comments}, who regards mixed strategies either ``as the distribution of the pure choices in the population'' or as ``a plan of action which is dependent on private information which is not specified in the model.'' In both interpretations, the player knows the outcome of her randomization over actions when forming her reference point. }

\textcolor{black}{In our existence proof, we focus on the case of finite type spaces. Hence, a mixed strategy for player $i$ is a mapping $\sigma_i:\Theta_i\rightarrow \Delta(\mathfrak{S}(\mathcal{S}))$, where $\Delta(\mathfrak{S}(\mathcal{S}))$ is the probability simplex over all ROLs. 
\textcolor{black}{Since the type space is finite, we can express function $\sigma_i$ as a vector of size $|\Theta_i|$.}}

{\color{black}
\begin{definition}
A mixed strategy profile $(\sigma_1,...,\sigma_n)$ is a choice-acclimating Bayes Nash equilibrium (CBNE) if for each student $i$ and each $\theta_i\in\Theta_i$, every ROL that is played with positive probability under $\sigma_i(\theta_i)$ is a CPE for student $i$, provided all other students follow the strategy prescribed in the mixed strategy profile.
\end{definition}

That is, in a CBNE each student $i$ anticipates other students' strategy profile $\sigma_{-i}$ and chooses a (pure strategy) CPE among the set her of CPEs. Since she is indifferent between any of the CPE, she may mix in her pure strategy choices according to the probabilities described by $\sigma_i(\theta_i)$.
\textcolor{black}{
In the appendix, we slightly adapt the existence proof by \cite{nash1951non} to reference-dependent preferences and asymmetric information. }

\begin{prop}\label{nash_exist}
Let $\Theta_i$ be finite for all students $i$. Then, a CBNE exists. 
\end{prop}
}


\subsubsection*{Misallocations with elite schools and the district-school bias}

We now \textcolor{black}{turn to the school-proposing deferred acceptance (DA) mechanism and} derive the district-school bias,
first explored by \citet{chen2006}.\footnote{\citet{hakimov2020} state the phenomenon that ``the district school (or safe school) is ranked higher in
the reported list than in the true preferences,'' and document its prevalence
in a wide range of experiments. In our two-school setting, it is equivalent to the small-school bias.}
For complex preference structures, there are numerous interdependencies, 
each giving rise to potential risks of instability. For now, we focus on the district-school bias and neglect 
other sources of misrepresentations, such as differences in preferences for schools within the set of desirable schools.

Suppose that there is a set $\mathcal E \subset \mathcal S$ of elite schools. 
Each school from this set is unambiguously preferred by each student over some 
safe outside option, the district school.
To simplify, we assume that all elite schools induce the same match utility $v>0$,
whereas the safe outside option induces a normalized utility of zero.
\textcolor{black}{If we consider this setting with standard preferences, DA implements the same allocation as TTC, i.e., the efficient allocation, in which the highest-score students attend the elite schools.}

Let a student's loss dominance $\Lambda_i$ be independently drawn from a common distribution with discrete support
$\{ \Lambda^0, \Lambda^1,  \dots,  \Lambda^l \}$. Since truthful reporting is a dominant strategy for any $\Lambda<1$, we can combine all loss dominance parameters in $[0,1]$ into $\Lambda^0$ and assume, without loss of generality, $\Lambda^0 =0$ and $\Lambda^1>1$.
By the following lemma, we can, without loss of generality, 
focus on just one elite school with capacity $q = \sum_{s \in \mathcal E} q_s<n$ 
instead of a set $\mathcal{E}$ of elite schools.

\vspace{-8pt}\begin{lem}\label{lem:elite}
For any belief on the attainability probabilities of elite schools, the best response is to either 
rank all elite schools adjacently (in any order) or no elite school before the district school.
\end{lem} \vspace{-8pt}

In equilibrium, a student's decision as to whether to apply to the elite school depends on 
her probability of attaining it, i.e., the probability that fewer than $q$ students of higher score apply there. 
Hence, the attainability probability is a function which is weakly increasing in 
her score $\omega$ and depends on the other students' reporting strategy $\sigma_{-i}$. 
Fixing $\sigma_{-i}$, the payoff function \eqref{eq:payoff-distance} implies that 
listing the elite school before the outside option is optimal for any $\omega>0$ if and only if
\begin{equation}\label{eq:elite}
  f(\omega)v-\Lambda f(\omega)(1-f(\omega))v\geq 0  \iff \Lambda \leq \frac{1}{1-f(\omega)}
 \iff  f(\omega) \geq 1 - \frac{1}{\Lambda}.
\end{equation}

We now turn to strategic interaction in the elite school problem. In general, strategic interaction between loss-averse agents is difficult to analyze and has to date only been sparsely studied. The following assumption puts more structure on school priorities, and gives a handle on strategic interaction in this mechanism. 

\begin{ass}[\textcolor{black}{Common priorities}] \label{ass:score}
$w_{i,s}=\omega_i$ for all $s \in \mathcal S$, \textcolor{black}{and
$\omega_i$ is identically and independently distributed according to some continuous distribution $G$ for all $i\in\mathcal I$.}
\end{ass}

Under common priorities, all schools evaluate students according to a one-dimensional priority score. For instance, the score may represent the result of a general assessment test, such as the SAT or GRE.\footnote{In many countries and cities, all schools use the same centralized score to rank students. See \\ \citet[Table 1]{fack2019}.}

A key observation in the analysis of strategic interaction \textcolor{black}{in DA} with \textcolor{black}{common priorities} is that a student's match outcome is only affected by the behavior of other students with a higher score. 
By construction of DA, a student will only be rejected by a school she proposes to if 
this school also has too many applications from a students with higher scores. \textcolor{black}{Hence, a student of score $\omega$ anticipates correctly the reporting behavior of potential students of higher scores, and derives her probability $f(\omega)$ to receive a match with the elite school if she applies. She then only applies if her loss parameter $\Lambda$ satisfies $\Lambda \leq \frac{1}{1-f(\omega)}$.} 


We say an equilibrium is essentially unique if
it is unique after imposing a rule for how students break ties
when they are indifferent between multiple ROLs.

\vspace{-8pt}\begin{lem}\label{lem:CBNE_elite}
In the elite school problem with common priorities, \textcolor{black}{DA has} an essentially unique choice-acclimating Bayesian Nash equilibrium. In this equilibrium, a student with 
loss dominance $\Lambda$ applies to the elite school if and only if her score is above some cutoff score $\overline\omega(\Lambda)\in(0,1)$, which is increasing in $\Lambda$. 
\end{lem} \vspace{-8pt}

The insight that pessimistic loss-averse students shy away from applying
has relevant ramifications for affirmative action policy beyond the scope of this simplified model.
While the cutoff attainability probability in \eqref{eq:elite} only depends on the score and the loss dominance,
beliefs can be additionally skewed which leads directly to unfavorable allocations.
For instance, if some students expect their scores to be lower than they are, because of (perceived) discrimination, these students may not apply to the elite school
although they would be assigned a seat in the \textcolor{black}{desired} allocation.
Consequently, \textcolor{black}{in this example} DA aggravates the discrimination by discouraging truthful revelation, and whether this discrimination is real or caused by underconfidence or doubts about how schools assess abilities, is irrelevant. Thus, downplaying the cost of discrimination when marginalized students do not rank 
discriminating schools in DA is inherently flawed in models incorporating EBLA, because the submitted ROLs may not reflect true preferences in equilibrium.

\section{Possible remedies}  \label{sec:remedy}

In keeping with Proposition \ref{prop:bounds}, misreporting in our model is inherently linked to beliefs on attainability and thereby on how a student expects to compare to others in terms of priority. While we modelled the source of uncertainty about this relative priority as uncertainty about other students' priorities $\mathbf{w}_{-i}$ at schools, the same holds true when it stems from uncertainty about how the schools assess own abilities $\mathbf{w}_i$. Notably, this uncertainty is unlikely to depend on market size. Moreover, our model abstracts away from other inherent sources of uncertainty in the market than the relative priority, and other students' preferences that may aggravate the problem, such as institutional uncertainties. In general, any regulation that may help a student to better assess his own standing and capacities of schools is likely to reduce students' uncertainty and hence the misrepresentation of preferences. In this section, we study how the choice of the mechanism itself may mitigate the problem of uncertainty and help to prevent justified envy in the equilibrium allocation.

\subsubsection*{Static mechanisms}

If we restrict to static matching mechanisms, we can provide a negative result.
A static mechanism is any mechanism which asks students about their preferences only once without providing feedback 
on other students' preferences. 
Formally, a static matching mechanism consists of reporting spaces $R=\times_{i\in\mathcal I} R_i$ for each student $i$ and an allocation function $o$, mapping reported profiles $\mathbf{r} = (r_i)_{i\in\mathcal I}\in R$ into allocations.

\vspace{-8pt}\begin{lem} \label{no_static}
For any distribution of preferences, there exists a static mechanism that implements the 
student-optimal stable allocation as CBNE for all realizations if and only if DA is truthful.
\end{lem} \vspace{-8pt}

Hence, if DA fails, there is no hope for remedies in the class of static mechanisms. Our argument is akin to the revelation principle for static mechanisms. Intuitively, if a student prefers to avoid the ex-ante risk that comes with the implementation of the student-optimal matching, she will not reveal her preferences under any such mechanism.

\subsubsection*{Sequential mechanisms and serial dictatorship}

Since uncertainty is the source of students' deviations, 
the use of sequential mechanisms may mitigate this problem. 
A sequential mechanism enables feedback between different rounds and, hence, 
has the ability to alter beliefs before eliciting preferences. The following example shows how misreporting may be diminished in the elite school problem when we move to sequential student receiving DA.

\begin{example}
Consider the elite school problem with two students and $q=1$.
In the unique stable matching, the student with the higher score is assigned to the elite school whereas the lower-score student is matched with a district school. The student receiving DA implements this allocation. Indeed, both schools propose to the stronger student, she accepts at the elite school, leaving the district school for the lower-score student.
Under the \textcolor{black}{static} student-proposing DA, however, students report their preferences truthfully only if \eqref{eq:elite} holds.
Consequently, if the score of both students is below $1-\nicefrac{1}{\Lambda}$, 
both students attend the district school, and the match outcome is neither stable nor student optimal.
\end{example}

\textcolor{black}{Note that in this example the stable match is unique, and the static student proposing DA fails to implement it. Hence, by Lemma 4, any static mechanism fails to implement a stable match in this case. This shows that stability cannot generally be assured by any static mechanism if students hold reference-dependent preferences.}

At first sight, it may seem surprising that the sequential use of information enables a designer to go beyond 
what is achievable with static mechanisms, as it seems to violate the fundamental insight of 
the revelation principle that any sequential mechanism has a static direct equivalent, 
see, e.g., \cite{myerson1979}.
In settings with dynamic information and EBLA, 
however, this revelation principle does not apply.
As students evaluate outcomes with respect to their beliefs, 
information endogenously affects their preferences over alternatives.

Unfortunately, moving to student receiving DA will not foster truthful behavior in general. It is well known that even with standard preferences DA is not truthful for the receiving side, as it implements the optimal stable allocation for the proposing side. In principle, the receiving side could coordinate on their preferred stable match by strategically rejecting all alternatives. \textcolor{black}{Hence, whether a change to sequential student receiving DA will improve students' match outcomes is in general ambiguous.}

Under common priorities, however, the stable match is unique, and the student receiving DA simplifies considerably.
When all schools have the same preferences, all schools approach the same student
in the first step. Then, this student is aware that she has the highest score among all students and
is immediately accepted at the school she selects.
All other schools are rejected and apply to the second-highest-score student
who is then aware that she is now the highest-score student of the unmatched population
and that she is assigned to her selected school with certainty, and so on.
In short, the student receiving DA simply becomes sequential serial dictatorship in which
homogeneous priority scores determine the order in which students pick their school.
Because each student determines her match with certainty, the dominance
of choosing the most preferred among the available options is obvious regardless of EBLA.

\citet{li2017} compares the outcome of the student proposing DA with the outcome of a 
sequential serial dictatorship mechanism in a lab experiment. 
He finds that while in the student proposing DA 36\% of games do not end in the stable outcome 
as induced by the dominant strategy, this rate drops to 7\% under sequential serial dictatorship. 
He explains this finding by the fact that, in contrast to serial dictatorship in this setting, 
student-proposing DA is not obviously strategy-proof (OSP). 
A mechanism is OSP if for the equilibrium strategy the worst outcome is still weakly better
than the best possible outcome from any alternative strategy, where 
the only outcomes considered are those that follow from the information sets where both strategies first diverge. 
Hence, dominance in an OSP mechanism may be easier to detect by agents with cognitive limitations.

\citet{ashlagi2018} show in their Example 1 that the sequential serial dictatorship mechanism 
is in general OSP when the proposing side has homogeneous preferences. 
However, they show that for general preferences
it is impossible to construct an OSP mechanism which always implements stable match outcomes. 
In Appendix \ref{sec:OSPvsEBLA}, we build on their Example 2 to demonstrate that 
a stable OSP mechanism may fail to induce stability with EBLA. 
This example sets the two concepts apart and suggests how to disentangle the explanations experimentally.
Our model conveys EBLA as an alternative explanation for the observed misrepresentations.
Rather than a mistake because of cognitive limitations,
we see them as the deliberate optimal choice of students who suffer from a behavioral bias.

\section{Conclusion}

We have identified a possible reason why students play \textcolor{black}{seemingly} dominated strategies in \textcolor{black}{static strategy-proof mechanisms}.
The truthful equilibrium in dominant strategies may not be a choice-acclimating personal equilibrium (CPE)
with dominant expectation-based loss aversion (EBLA).
In other contexts, evidence consistent with dominant EBLA has been found 
in numerous experimental and field studies.
The notion that students forgo small chances to get into preferred schools to avoid disappointment
is therefore plausible.
Indeed, the costly deviations from the dominant truthful strategy are most pervasive
among low- and intermediate-priority students who want to get into competitive programs.
Our theoretical predictions fit this pattern in experimental and field data, and we also
provide a formalized framework for the pervasive district-school and small-school biases.
Our characterization of optimal play in Proposition \ref{prop:order} and \ref{prop:bounds}
is testable.

The extensive evidence of dominated play calls into question
the identification strategy to treat reported preferences as truthful.
Regarding affirmative action this insight is important, because the observation that people of certain demographics do not reveal a preference for certain schools in \textcolor{black}{a static strategy-proof mechanism} does not imply that they do not want to go there.

\section*{Appendix}
\setcounter{subsection}{0}
\renewcommand{\thesubsection}{A.\Roman{subsection}}

\subsection{Relation to \citet{dreyfuss2019}} \label{sec:other-paper}

Similar to our paper,  \citet{dreyfuss2019} find that EBLA can explain non-truthful ROLs.
In their reduced form dynamic framework \`{a} la \citet{koszegi2009}, 
students enter the decision problem with a reference point given by the outside option, 
whereas in our decision problem students already anticipate the choices ahead of them, 
which is reflected in their reference point. Moreover,  \citet{dreyfuss2019} consider an extra period where uncertainty is resolved, which gives rise to additional gain-loss utilities. The essential intuition for how students use 
manipulations to shield off potential disappointment is, however, similar in both models.

We take the stylized approach that gains and losses are assigned when comparing to the value of 
other potential outcomes (narrow bracketing).  \citet{dreyfuss2019} take the opposite approach as they consider each school 
in a separate consumption dimension and assign gains and losses separately for each school. 
The reality is certainly somewhere in-between, as schools may be comparable in some aspects but not in others. 
We choose our modeling approach to draw a clear comparison to the existing experimental literature, 
where stakes are simply money, and values are hence fully comparable between 
schools.

The uncertainty in \citet{dreyfuss2019} stems from iid shocks on how individual schools assess 
a student's abilities with respect to exogenously given school standards. 
This reduced form approach has two implications. First, it leaves no scope for strategic interaction 
between students and, hence, cannot be operationalized in the matching context immediately. Second, it implies that attainability probabilities are independent between schools, 
which is not the case in our model, not even under Assumption \ref{ass:score} and independently drawn scores. Assuming independent attainability probabilities would exclude the existence of ``popular schools," and we want to explain why students shy away from applying to such schools.

From the first theoretic insight that under EBLA there is scope for strategic misrepresentations, both papers proceed quite complementarily. While \citet{dreyfuss2019} comprehensively reevaluate 
the experimental data in \citet{li2017} in light of EBLA, 
we delve deeper into its theoretical implications, 
and analyze the set of rationalizable strategies, strategic interaction, and evaluate remedy mechanisms.

\subsection{OSP versus EBLA} \label{sec:OSPvsEBLA}

This section illustrates the distinction of the notion of robustness against EBLA and the concept of OSP. 
We start with the observation that robustness against EBLA does not imply that a mechanism is OSP. 
By Proposition \ref{prop:bounds}, all students will report truthfully in DA if their probability $p_1$ 
that their favorite school is attainable is sufficiently large.
This condition certainly does not imply that DA is OSP.
For example, ranking the favorite school 1 second can yield a match with 1 as a best case, whereas 
the true ROL can lead to a worse match.

Building on Example 2 in \citet{ashlagi2018}, 
we now provide an example of acyclical preferences and an OSP mechanism 
which always implements the student-optimal stable matching with standard preferences, 
but fails to do so if students exhibit EBLA.
There are two students, $I = \{A,B\}$ and two schools, $S=\{1,2\}$. School 1 prefers student $A$ over $B$, whereas school 2 prefers student $B$ over $A$. Conversely, student $A$ prefers school 1 over
school 2 with probability $(1-\epsilon)$, and student $B$ prefers school 2 over school $A$ with probability of
$(1-\epsilon)$ for some small $\epsilon>0$.
Here, DA is not OSP. 
For instance, if student $A$ prefers school $2$ truth-telling is not obviously dominant 
as the true ROL $(2,1)$ may result in a match with school $1$, 
whereas ROL $(1,2)$ may result in a match with $2$ with positive probability.

\citet[Figure 2]{ashlagi2018} propose the following OSP sequential mechanism.
First, student $A$ is asked whether she prefers $1$ or $2$. If she prefers $1$, she is assigned
to $1$ and $B$ is assigned to $2$. If she prefers $2$, $B$ is asked for her preferences
which then determine the match outcome.
Because $B$ determines the match with certainty whenever she is asked, revealing her true preferences
is an obviously dominant strategy (and a sequential CPE at this final decision node).
If $A$ prefers school 1, deviating yields her either
$v_{A,1}$ or $v_{A,2}$ instead of a certain payoff $v_{A,1}>v_{A,2}$ 
such that the truth is both a sequential CPE and an obviously dominant strategy.

If $A$ prefers school 2, misrepresenting yields her a sure payoff of $v_{A,1}<v_{A,2}$ and
being truthful leads to a lottery over $v_{A,2}$ and $v_{A,2}$.
Because even the worst outcome from being truthful is not worse than the best (only)
outcome from deviating, the truth is an obviously dominant strategy, making
the mechanism OSP.
However, for any $\Lambda>1$, truth-telling is not a
sequential CPE for student $A$ if $\epsilon<1-\nicefrac{1}{\Lambda}$.
Hence, because of the uncertainty effect, even OSP mechanisms can fail to have a truthful sequential CPE.

\subsection{Proofs} \label{sec:proofs}

\begin{proof}[Proof of Lemma \ref{lem:acc}]
Take \textcolor{black}{any static strategy-proof mechanism that only collects ROLs of all players simultaneously and fix} an arbitrary ROL \textcolor{black}{of student $i$. L}et $s$ be the highest-ranked attainable school in it.

Suppose that under this ROL the student is matched with $s'$ ranked before $s$. But then, since $s'$ is unattainable (i.e., she would not get in if ranked first), she would prefer this ROL over her true ROL if $s'$ was her most preferred school, a contradiction to strategy-proofness.

Suppose that under this ROL she is matched with $s''$ ranked behind $s$. 
But then, if the ROL was true, she would prefer a match with $s$ over $s''$, and ranking $s$ first
would achieve this match, again a contradiction to strategy-proofness.
\end{proof}

\begin{proof}[Proof of Proposition \ref{prop:order}]
a)
We start with a practical lemma which identifies
when flipping two adjacently ranked schools in an ROL is profitable. \textcolor{black}{Fix an arbitrary attainability distribution $P$.}
Consider two otherwise identical ROLs swapping two adjacently ranked schools $x<y$, i.e.,
two ROLs $(...,x,y,...)$ and $(...,y,x,...)$. Let the former induce lottery $F =(f_s)_{s \in \mathcal S}$
and the latter induce lottery $\underline F =(\underline f_s)_{s \in \mathcal S}$ \textcolor{black}{as given by \eqref{eq:probs}}, and let
$\varepsilon$ denote the probability of 
$x$ and $y$ being attainable but no school which is ranked before.

\vspace{-8pt}\begin{lem}\label{flip}
	$U(\cdot, F)\geq U(\cdot, \underline F)$ if and only if
	\begin{equation}\label{eq:aux}
\frac{\varepsilon}{\Lambda}\geq \varepsilon\left(- \sum\limits_{s=1}^{x} f_s+ \varepsilon
+\sum\limits_{s=x+1}^{y-1} f_s \: \frac{v_x + v_{y} - 2 v_s}{v_x- v_{y}} + \sum\limits_{s=y}^{m} f_s \right)
	\end{equation} 
	with equality in \eqref{eq:aux} only in the case of indifference.
\end{lem} \vspace{-8pt}

\begin{proof}[Proof of Lemma \ref{flip}]
By  \eqref{eq:payoff-distance}, we have 
\begin{equation}\label{eq:utility_diff}
U(\cdot, F)-U(\cdot, \underline F)=\sum_{s=1}^m(f_s-\underline f_s)v_s-\Lambda\sum_{s=1}^m\sum_{r=s+1}^m(f_sf_r-\underline f_s\underline f_r)(v_s-v_r).
\end{equation}
For the matching probabilities $\underline{f}_s$ of ROL $(...,y,x,...)$,
it must be that $f_s=\underline{f}_s$ for $s\neq x,y$ and $\underline{f}_x=f_x-\varepsilon$, $\underline{f}_y=f_y+\varepsilon$, with $\varepsilon\geq 0$. Hence, if we split the second sum over $s$ into five summands, we obtain  \small
\begin{align*}
	&\sum_{s=1}^m\sum_{r=s+1}^m(f_sf_r-\underline f_s\underline f_r)(v_s-v_r) \\
	=& \sum_{s=1}^{x-1} \left[(f_sf_x-f_s \underline f_x)(v_s-v_x)+ (f_sf_y-f_s \underline f_y)(v_s-v_y)\right]+ \sum_{r=x+1}^m(f_xf_r-\underline f_x\underline f_r)(v_x-v_r)\\
	&+\sum_{s=x+1}^{y-1}(f_sf_y-f_s\underline f_y)(v_s-v_y)+\sum_{r=y+1}^m(f_yf_r-\underline f_yf_r)(v_y-v_r)+0\\
	= & \sum_{s=1}^{x-1} \left[ f_s\varepsilon (v_s-v_x)- f_s\varepsilon(v_s-v_y)\right]  + \sum_{r=x+1}^m\varepsilon f_r (v_x-v_r) + (\underline f_xf_y-\underline f_x\underline f_y)(v_x-v_y)\\
	&+\sum_{s=x+1}^{y-1}f_s(-\varepsilon)(v_s-v_y)+\sum_{r=y+1}^m(-\varepsilon)f_r (v_y-v_r)\\
		= & -\sum_{s=1}^{x-1}\varepsilon f_s (v_x-v_y)  + \sum_{r=y}^m\varepsilon f_r (v_x-v_r) + (f_x-\varepsilon)(-\varepsilon)(v_x-v_y)\\
	&+\sum_{s=x+1}^{y-1}\varepsilon f_s(v_x-2v_s+v_y)+\sum_{r=y+1}^m(-\varepsilon)f_r (v_y-v_r)\\
	=& \varepsilon(v_x-v_y) \left(- \sum_{s=1}^{x} f_s +\varepsilon +\sum_{s=y}^m f_s + \sum_{s=x+1}^{y-1} \frac{v_x-2v_s+v_y}{v_x-v_y} \right)
\end{align*}
\normalsize	
Since the difference in classical utility satisfies $\sum_{s=1}^m(f_s-\underline f_s)v_s=\varepsilon (v_x-v_y)$, we have $U(\cdot, F)-U(\cdot, \underline F)\geq 0$ if and only if
\[\varepsilon (v_x-v_y)\geq \Lambda \varepsilon(v_x-v_y) \left(- \sum_{s=1}^{x} f_s +\varepsilon +\sum_{s=y}^m f_s + \sum_{s=x+1}^{y-1} \frac{v_x-2v_s+v_y}{v_x-v_y} \right).\]
Dividing by $\Lambda(v_x-v_y)>0$ yields the result for the inequality. For the statement about indifference, replace all inequalities with equality.
\end{proof}

We now prove the next auxiliary lemma by contradiction.
\vspace{-8pt}\begin{lem}\label{lem:characterize}
\textcolor{black}{
Consider ROL $(s_1,s_2,\dots,s_m)$, and let it be strictly optimal. Alternatively, suppose the attainability distribution has full support and let it be weakly optimal.
Then, $s_k<s_\ell$ for $k>\ell$ imply that for all $s_i<s_j\leq s_k,$ $i>j$. That is, if the (weakly) optimal ROL ranks school $b$ after school $c$ for $b<c$, then it ranks the schools $1,...,b-1,b$ in decreasing order.}
\end{lem} \vspace{-8pt} 
\begin{proof}[Proof of Lemma \ref{lem:characterize}]
Suppose that for some $1\leq a < b <c\leq m$,
the \textcolor{black}{respective} optimal ROL ranks $b$ behind $c$ but $a$ before $b$. 
Let $c$ be the least preferred school, i.e., the one with the
highest index, for which such a triple exists in this ROL. 
Given this $c$, select $b$ and $a$ such that $a$ is the lowest-index school,
i.e., the most preferred one, satisfying the requirement.

Since $a$ is ranked before $b$, the optimal ROL has one of the following forms:
\begin{enumerate}[i)]
  \item $(...,a,...,c,...,b,...)$
  \item $(...,c,...,a,...,b,...)$
\end{enumerate}
We make first considerations for both cases.

i) 
Since, by assumption, $a$ is the lowest-index school ranked before $b$, the list must be increasing from $a$, and eventually decreasing (possibly at $b$) to a number above $a$. Call $\overline x$ the first school
      where the list starting from $a$ has decreased, \textcolor{black}{such that $\overline x$ is the first school after $a$ which has a lower index than its preceding school in the list.} Now, by choosing $\underline x$ appropriately in the list between $a$ and $\overline x$, we obtain in the optimal ROL a sequence $(...,\underline x,\underline
      y,...,\overline y,\overline x,...)$ (with possibly $\underline y=\overline y$), which is increasing from $\underline x$ to $\overline y$ and satisfies $\underline x< \overline x< \underline y \leq \overline y$.

ii) 
Since, by assumption, $c$ is the highest-index school for which there exists $b$ and $a$ with $b$ ranked 
behind $c$ but $a$ before $b$, the list must be decreasing from $c$, 
but eventually increasing (possibly immediately after $a$) to a number below $c$. 
Call $\underline y$ the first school after $c$ where the list is increasing. 
Now, by choosing $\overline y$ appropriately in the list between $c$ and $\underline y$, we obtain in
      the optimal ROL a sequence $(...,\overline y,\overline x,...,\underline x,\underline y,...)$ (with possibly $\overline x=\underline x)$, which is decreasing from $\overline y$ to $\underline x$ and satisfies
      $\underline x\leq \overline x< \underline y < \overline y$.

The rest of the proof is identical for both cases.

Since the ROL is supposed to be \textcolor{black}{(weakly)} optimal, it must be \textcolor{black}{(weakly)} preferred to an otherwise
equivalent ROL that swaps $\overline x$ and $\overline y$.
Let $f_s$ be the matching probabilities as induced by the optimal ROL, and let $\overline{f}_s$ be the matching probabilities as induced by the (otherwise identical) ROL
that flips $\overline x$ and $\overline y$.
By \textcolor{black}{Lemma \ref{lem:acc}}, we obtain $f_s=\overline f_s$ for all $s\neq \overline x,\overline y$, and
\ean{
\overline f_{\overline x} &= f_{\overline x} + \overline \varepsilon \quad \mbox{and} \quad
&\overline f_{\overline y} = f_{\overline y} - \overline \varepsilon \label{eq:probaio},
}
where $\overline\varepsilon$ is the probability that $\overline x$ and $\overline y$ are attainable, but 
any school ranked before $\overline x$ and $\overline y$ in the optimal ROL is not.
By \textcolor{black}{either} strict optimality \textcolor{black}{or the full-support assumption}, $\overline \varepsilon >0$.

Hence, by Lemma \ref{flip} \small
\eq{ \label{Loverline}
\begin{aligned}
\frac{1}{\Lambda}&\textcolor{black}{\leq} - \sum\limits_{s=1}^{\overline x} \overline f_s+ \overline\varepsilon
+\sum\limits_{s=\overline x+1}^{\overline y-1} \overline f_s \: \frac{v_{\overline x} + v_{\overline y} - 2 v_s}{v_{\overline x}- v_{\overline y}} + \sum\limits_{s=\overline y}^{m} \overline f_s \\
&= - \sum\limits_{s=1}^{\overline x}  f_s
+\sum\limits_{s=\overline x+1}^{\overline y-1} f_s \: \frac{v_{\overline x} + v_{\overline y} - 2 v_s}{v_{\overline x}- v_{\overline y}} + \sum\limits_{s=\overline y}^{m}  f_s  - \overline\varepsilon
\end{aligned}
}  \normalsize

Similarly, the ROL must be \textcolor{black}{weakly} preferred to an otherwise
equivalent ROL that swaps $\underline x$ and $\underline y$.
Hence, by Lemma \ref{flip} 
\eq{ \label{Lunderline}
\frac{1}{\Lambda}\textcolor{black}{\geq}
- \sum\limits_{s=1}^{\underline x} f_s+ \underline\varepsilon
+\sum\limits_{s=\underline x+1}^{\underline y-1} f_s \: \frac{v_{\underline x} + v_{\underline y} - 2 v_s}{v_{\underline x}- v_{\underline y}} + \sum\limits_{s=\underline y}^{m} f_s 
}

Both inequalities can only simultaneously hold if the right-hand side of (\ref{Loverline}) is \textcolor{black}{weakly} larger than the right-hand side of (\ref{Lunderline}), which we bring to a contradiction. Since $-\overline\varepsilon<0<\underline\varepsilon$ it suffices to show that for each $s$ the respective summand in (\ref{Loverline}) is \textcolor{black}{weakly} smaller than in  (\ref{Lunderline}).

For $s\in[\underline x +1, \overline x]$, we have $-1= \frac{v_{\underline x} + v_{\underline y} - 2 v_s}{v_{\underline x}-v_{\underline y}}+  2\frac{ v_s-v_{\underline x} }{v_{\underline x}-v_{\underline y}}< \frac{v_{\underline x} + v_{\underline y} - 2 v_s}{v_{\underline x}-v_{\underline y}}$, since $v_{\underline x}>v_s>v_{\underline y}$.
For $s\in[\overline x +1, \underline y-1]$, we have  
\[ \frac{v_{\overline x} + v_{\overline y} - 2 v_s}{v_{\overline x}- v_{\overline y}}  = \frac{1}{1+\frac{2v_s-2v_{\overline y}}{v_{\overline x}+v_{\overline y}-2v_s}}< \frac{1}{1+\frac{2v_s-2v_{\underline y}}{v_{\underline x}+v_{\underline y}-2v_s}}=\frac{v_{\underline x} + v_{\underline y} - 2 v_s}{v_{\underline x}- v_{\underline y}} ,\]
where the inequality follows since $v_{\underline x}>v_{\overline x},~v_{\underline y}>v_{\overline y}$, and the term is increasing in both, $v_{\overline x}$ and $v_{\overline y}$.
For $s\in[\underline y, \overline y-1]$, we have  $\frac{v_{\overline x} + v_{\overline y} - 2 v_s}{v_{\overline x}- v_{\overline y}} =1+ 2\frac{v_{\overline y}-v_s}{v_{\overline x}- v_{\overline y}}<1$ since $v_{\overline x}>v_s>v_{\overline y}$.
\end{proof}

Now, a) of the proposition follows immediately.
That schools preferred over \textcolor{black}{the top-ranked school} $b$ are optimally ranked in decreasing order is just Lemma \ref{lem:characterize}. 
An ROL not ranking schools worse than a \textcolor{black}{the top-ranked school} $a$ in increasing order
would rank some $a<b<c$ in the form of $(a,...,c,....,b,...)$.
As seen by the contradiction of case (i) in the proof of  Lemma \ref{lem:characterize},
this cannot be \textcolor{black}{strictly} optimal \textcolor{black}{and cannot be weakly optimal if the attainability distribution has full support}.

We now establish assertion b). First, by Proposition \ref{prop:bounds} the truthful ROL is always strictly optimal for a sufficiently large attainability probability $p_1$ at the most preferred school.\footnote{Note, that this is not a circular argument, as Proposition \ref{prop:bounds} only builds on Proposition \ref{prop:order}a.}  Hence, we can focus on non-truthful, top-\textcolor{black}{rank} monotone ROLs.

We say that  the joint distribution of attainability probabilities has full support if for any subset of schools without the safe school $m$ the probability of these schools being attainable and these schools only is strictly between zero and one. We now show by induction over the number of schools that for any non-truthful top-\textcolor{black}{rank} monotone ROL there exists an \textcolor{black}{attainability distribution $P$} with full support on the distribution of attainability probabilities such that the ROL is strictly optimal.

For the base case $\mathcal S_2=\{1,2\}$, Lemma \ref{flip} establishes that $(2,1)$ is strictly optimal for any $v_1>v_2$ if $p_1$ satisfies
\[\frac{p_1}{\Lambda}<p_1\left(-p_1+1\right) \quad\Leftrightarrow \quad 0<p_1<1-\frac 1 \Lambda.\]

For the induction step suppose that the statement holds for any set of $m-1$ schools. Let $L$ be an arbitrary non-truthful top-\textcolor{black}{rank} monotone list for the set $\mathcal S_m=\{1,...,m\}$.  Since any such ROL must rank school 1 behind school 2, list $L$ must be of the form $([a],2,[b],1,[c])$, where $[a]$, $[b]$, and $[c]$ stand for some (potentially empty) ordered subsets of schools $\{3,4,\dots,m\}$.

By induction assumption, for the school set  $\mathcal S_{m-1}=\{\tilde 2,3,...,m\}$ there is some \textcolor{black}{$\widetilde P$} with full support such that $\widetilde L =([a],\widetilde 2,[b],[c])$ is strictly optimal.
Intuitively, we now construct the environment for set $\mathcal S_m$ that makes $L$ strictly optimal by splitting school $\widetilde 2$ into two schools, $1$ and $2$. Formally, we define  $v_1>v_{\tilde 2}$ and $v_s=v_{\tilde s}$ for all $s\geq2$. Attainability probabilities are defined exactly as on the set $\mathcal S_{m-1}$ with the additional assumption that whenever school $\tilde 2$ is attainable school 1 but not 2 is attainable with probability $\epsilon$, school $2$ but not $1$ with $1-2\epsilon$, and both schools $1$ and $2$ with probability $\epsilon$.\footnote{More formally, for any $E\in \Pi_{s=3}^m A_s$ we define $P (A_1=A_2=1, E) = \epsilon \widetilde P (A_{\widetilde 2}=1, E), \quad
P (A_1=1,A_2=0, E) = \epsilon \widetilde P (A_{\widetilde 2}=1,  E),
P (A_1=A_2=0, E) = \widetilde P (A_{\widetilde 2}=0,  E), \quad
P (A_1=0,A_2=1, E) = (1-2\epsilon) \widetilde P (A_{\widetilde 2}=1, E)$.} Note that the resulting distribution has full support. We need to show that for suitable choices of $v_1$ and $\epsilon$ list $L$ becomes strictly optimal. Foreshadowing the structure of the argument, we split the remaining proof into proving the following three claims.
\textit{
\begin{itemize}
    \item[(i)] If $p_1=\epsilon p_{\widetilde 2} < \frac{1-1/\Lambda}{2}$, then 
    there is some $\overline v>v_2$ 
such that for $v_1=\overline v$ any optimal ROL ranks
school $1$ after outside option $m$.
\item[(ii)] For sufficiently small $\epsilon$ and $v_1\leq\overline v$, any strictly optimal ROL  of set $\mathcal S_m$ ranks schools $\{2,\dots,m\}$ in the order $([a],2,[b],[c])$.
\item[(iii)] For sufficiently small $\epsilon$ there is some $v_1\in (v_2,\overline v]$ such that $([a],2,[b],1,[c])$ is strictly optimal.
\end{itemize}
}
Proof of Claim (i):
For any ROL that does not rank school 1 last, consider Lemma \ref{flip} for a switch of ranks between school $x=1$ and school $y$ ranked directly after it. Define
\[\alpha(v_1,v_s,v_y)=\frac{v_1+v_y-2v_s}{v_1-v_y}=1-2\frac{v_s-v_y}{v_1-v_y}\]
and let $\underline \alpha(v_1)=\min_{2\leq s<y\leq m} \alpha(v_1,v_s,v_y)$. Note that $\underline\alpha(v_1)$ is strictly increasing in $v_1$ with $\underline\alpha(v_1)\in(-1,1)$ for all $v_1$, and $\lim_{v_1\to\infty}\underline\alpha(v_1)=1$.

Considering the swap of $1$ and $y$, we have \small
\begin{align*} - f_1+ \varepsilon
 +\sum\limits_{s=2}^{y-1} f_s \: \alpha(v_1,v_s,v_y) + \sum\limits_{s=y}^{m} f_s > -f_1 + \underline\alpha(v_1)\sum_{s=2}^m f_s > 1-2f_1 + (\underline\alpha(v_1)-1),
 \end{align*} \normalsize
 as $\sum_{s=2}^m f_s=(1-f_1)$ and $\underline \alpha(v_1)<1$.
 If $p_1<\overline p\equiv  \frac{1-1/\Lambda}{2}$, then $f_1\leq p_1$ implies
 \[ 1-2f_1 + (\overline\alpha(v_1)-1)>1-2\overline p +(\overline\alpha(v_1)-1)=\nicefrac{1}{\Lambda}+(\overline\alpha(v_1)-1).\]
 Since the inequality is strict and $\lim_{v_1\to\infty} (\overline\alpha(v_1)-1)=0$, 
 there is some $\overline v$, for which swapping 1 and $y$ is profitable by Lemma \ref{flip} as $1-2\overline p +(\overline\alpha(v_1)-1)>\nicefrac{1}{\Lambda}$
 and $\varepsilon>0$ by full support. Hence, it is always profitable to switch ranks of 1 and $y$ for any school $y$ ranked behind it, and, by iteration,
 ranking $1$ before $m$ is strictly suboptimal for $v_1= \overline v$. 
 
 Proof of Claim (ii):
 Intuitively, for small $\epsilon$ the position of school 1 is unsubstantial for expected utility such that the optimality of the order $([a],2,[b],[c])$ follows from the strict optimality of $([a],\tilde 2,[b],[c])$.
 
 More formally, let $L'=([d],1,[e])$ be an arbitrary ROL of school set $\mathcal S_m$. We calculate a bound on the utility difference between  $L'$ with attainability distribution $P$ and ROL $\tilde L'=([d],[e])$ under school set $\mathcal S_{m-1}$ and attainability distribution $\widetilde P$ (indentifying school 2 with school $\tilde 2$).

Let $p_1=\epsilon p_{\tilde 2}$ be the unconditional probability of school 1 being attainable. Denote with $(f'_s)_{s_\in\mathcal S_{m}}$ the match probabilities of list $L'$ and with $(\tilde f'_s)_{s_\in\mathcal S_{m-1}}$ the match utilities of list $\tilde L'$. Let further be $\Delta U([d],[e])$ the absolute value of the utility difference between $L'$ and $\tilde L'$. Employing triangle inequality and the fact that the match probabilities with schools $2,...,m$ each differ by less then $p_1$ between ROLs $L'$ and $\tilde L'$, we obtain  by  \ref{eq:payoff-distance}
\small
\begin{align*}
\Delta U([d],[e])&\leq 
\sum_{s=1}^m |\tilde f'_s-f'_s| v_s +\Lambda \sum_{s=1}^m \Big( \left| f'_s \sum_{r=s+1}^m f'_r (v_s-v_r) - \tilde f'_s \sum_{r=s+1}^m \ f'_r(v_s-v_r)\right| \\
&+ \tilde f'_s\sum_{r=s+1}^m p_1 (v_s-v_r) \Big) \\
&<\sum_{s=1}^m p_1 \overline v+ \Lambda \sum_{s=1}^m \left( p_1\overline v\sum_{r=s+1}^m f'_r +  \tilde f_s'(m-s-1) p_1 \overline v \right)\\
&<m p_1 \overline v +\Lambda 2 m p_1 \overline v =\epsilon p_{\tilde 2} m \overline v (1+2\Lambda).
\end{align*}
\normalsize
As $([a],\tilde 2 ,[b],[c])$ is strictly optimal for $\mathcal S_{m-1}$ and $\tilde P$, let $C$ be the utility difference to the second best ROL of $\mathcal S_{m-1}$. Let $\epsilon< \frac{C}{2p_{\tilde 2} m \overline v (1+2\Lambda)}$. Hence, whenever the order of schools $\{2,...,m\}$ in $([d],[e])$ is different to the order in $([a],2,[b],[c])$ we know that the utility of $([a],2,1,[b],[c])$ exceeds the utility of $([d],1,[e])$ by more that 
\[ -\Delta U([a],2,[b],[c])+ C - \Delta U([d],[e]) > - C/2 +C  -C/2 = 0,\]
which shows that any optimal ROL of $n$ schools ranks schools $2,...,m$ in the order $([a],2,[b],[c])$.

 Proof of Claim (iii): By top-\textcolor{black}{rank} monotonicity, $([b],[c])$ is of form $(k,k+1,...,m)$ with  $k\geq 3.$. Also by top-\textcolor{black}{rank} monotonicity, it is suboptimal to rank 1 behind 2. We need to show that for sufficiently small $\epsilon$ any position of school 1 in $(k,k+1,...,m)$ can be achieved as strictly optimal by choosing $v_1$ appropriately. 
We employ Lemma \ref{flip} to consider a swap of school $x=1$ and school $y$ ranked directly behind 1 (if any).
For fixed match probabilities $f_1,...,f_m$, the term in brackets on the right-hand side of \eqref{eq:aux} is strictly decreasing in $y$ for any $v_1$,  since $\alpha(v_1,v_s,v_y)<1$ is increasing in $v_y$ and $v_y$ is decreasing in $y\in\{k,k+1,...,m\}$. Moreover, the amount by which this term decreases depends continuously on $v_1$, and hence the term attains a minimum strictly larger than zero at some $v_1 \in [v_2,\overline v]$. Since the match probabilities $f_s$ and $\varepsilon$ can differ by at most $p_1=\epsilon \widetilde p_{\widetilde 2}$ between two different ROLs it follows that the term in the brackets of \eqref{eq:aux} is also strictly decreasing in $y\in\{k,k+1,...,m\}$ for all $v_1\in[v_2,\overline v]$ for sufficiently small $\epsilon.$ Hence, an order $([a],2,k,...,y-1,1,y,...,m)$ is optimal if and only if there is a $v_1$ for which the term in the brackets of \eqref{eq:aux} is strictly larger than $\frac 1 \Lambda$ for the adjacent swap of $1$ and $y-1$, but strictly smaller than $\frac 1 \Lambda$ for the adjacent swap of $1$ and $y$. Since for each $y$ the term in the brackets of \eqref{eq:aux} is continuously increasing in $v_1$ and decreasing in $y$, such $v_1\in(v_2,\overline v]$ exists by the intermediate value theorem for all $y$ if $([a],2,1,k,...m)$ is optimal for $v_1$ sufficiently close to $v_2$ and $([a],2,k,...,m,1)$ is optimal for $v_1=\overline v$. The latter has been shown in Claim (i). For $v_1$ sufficiently close to $v_2$ the optimality of $L=([a],2,1,k,k+1,...,m)$ follows from the optimality of $\widetilde L=([a],\widetilde 2, k,k+1,...,m)$ for schools $\mathcal S_{m-1}$: By Lemma \ref{flip}, optimality of $\tilde L$ implies for this list
\[\frac{1}{\Lambda}>-f_{\tilde 2}+\varepsilon +\sum_{s=3}^{k-1}f_s \frac{v_{\tilde 2}+v_k-2v_s}{v_{\tilde 2}-v_k} + \sum_{s=k}^m f_s. \]
Since the match probabilities for schools 3,...,m are the same in $([a],2,1,k,k+1,...,m)$ and $\tilde L$, the match probabilities $f_1$ and $f_2$ satisfy $f_1+f_2=f_{\tilde 2}$, and for $v_2=v_{\tilde 2}$, we have 
\begin{align*}
\frac{1}{\Lambda}&>-f_1+\frac{v_2+v_k-2v_2}{v_2-v_k}f_2+\epsilon\varepsilon +\sum_{s=3}^{k-1}f_s \frac{v_{ 2}+v_k-2v_s}{v_{ 2}-v_k} + \sum_{s=k}^m f_s \\
&= \lim_{v_1\to v_2} \left(-f_1+\epsilon\varepsilon+ \sum_{s=2}^{k-1}  \frac{v_{1}+v_k-2v_s}{v_{1}-v_k}f_s + \sum_{s=k}^m f_s  \right),
\end{align*} 
which shows that it is unprofitable to swap $1$ and $k$ in $([a],2,1,k,k+1,...,m)$ when $v_1$ is sufficiently close to $v_2$ . This concludes the proof of Claim (iii).
\end{proof}

\begin{proof}[Proof of Proposition \ref{prop:bounds}]
1.
 Suppose, by way of contradiction, that truth-telling is suboptimal. By Proposition \ref{prop:order}a,
  this implies that there is an optimal ROL which does not rank school 1 first. Let $\overline F=(\overline f_1,...,\overline f_m)$ be the lottery induced by thatl ROL. Let $F=(f_1,...,f_m)$ be the lottery induced by the ROL which ranks school 1 first and all other schools in the same order as the optimal ROL. Let further $\varepsilon_i= f_i-\overline f_i$ be the shifts in probability when choosing $F$ instead of $\overline F$. Evidently, $\varepsilon_i\leq0$ for all $i\geq 2$ and $\varepsilon_1=-\sum_{i=2}^m \varepsilon_i>0$, where the strict inequality comes from the fact that school 1 is not exclusive. To find a contradiction to the optimality of $\overline F$, we show that $U(\cdot,F)>U(\cdot,\overline F)$. By \eqref{eq:payoff-distance},
  \small
\begin{align*}
U(\cdot,F)-U(\cdot,\overline F)&= \sum_{s=1}^m f_s \left[ v_s-\Lambda \sum_{r=s+1}^m f_r(v_s-v_r) \right] - \sum_{s=1}^m \overline f_s \left[ v_s-\Lambda \sum_{r=s+1}^m \overline f_r(v_s-v_r) \right] \\
&\geq \sum_{s=1}^m f_s \left[ v_s-\Lambda \sum_{r=s+1}^m f_r(v_s-v_r) \right] - \sum_{s=1}^m \overline f_s \left[ v_s-\Lambda \sum_{r=s+1}^m  f_r(v_s-v_r) \right] \\
&=  \left(- \sum_{i=2}^m \varepsilon_i\right) \left[ v_1-\Lambda \sum_{r=2}^m f_r(v_1-v_r) \right]  +  \sum_{s=2}^m \varepsilon_s \left[ v_s-\Lambda \sum_{r=s+1}^m f_r(v_s-v_r) \right] \\
&= \sum_{s=2}^m \varepsilon_s \left[ v_s-v_1 +\Lambda \sum_{r=2}^s f_r(v_1-v_r) + \Lambda \sum_{r=s+1}^m f_r(v_1-v_s)  \right] \\
&> \sum_{s=2}^m \varepsilon_s \left[ v_s-v_1 +\Lambda \sum_{r=2}^s f_r(v_1-v_s) \right] \\
&= \sum_{s=2}^m \varepsilon_s \left[ v_s-v_1 +\Lambda (1- p_1) (v_1-v_s) \right] >0,
\end{align*}
\normalsize
where the last inequality exploits that, by assumption, $p_1>1-1/\Lambda$.

2.
Suppose, by way of contradiction, that truth-telling is strictly optimal. 
Hence, $(1,2,...,m) \succ_i (2,1,...,m)$, and by Lemma \ref{flip}
  \[\frac{\varepsilon}{\Lambda}\geq \varepsilon\left(-f_1+\varepsilon+\sum_{s=2}^{m}f_s\right),\]
where $\varepsilon$ is the probability that the student is attainable at schoo 1 and school 2.
Since school 1 is not exclusive, we have $\varepsilon>0$, and hence
  \[\frac{1}{\Lambda}\geq -f_1+\varepsilon+\sum_{s=2}^{m}f_s=-f_1+\varepsilon+(1-f_1)>1-2f_1=1-2p_1,\]
  which can be rearranged to $p_1 > 0.5\left(1-\nicefrac{1}{\Lambda}\right)$, a contradiction.
\end{proof}


{\color{black}
\begin{proof}[Proof of Proposition \ref{nash_exist}]
A mixed strategy $\sigma_i$ can be identified as an element of the space $(\Delta(\mathfrak{S}(\mathcal{S}))\big)^{|\Theta_i|}$, which is a compact subset of $\mathbb R^{ m!\cdot |\Theta_i|}$. It is convenient to think about such an element as a matrix which specifies for each ROL $r\in\mathfrak{S}(\mathcal{S})$ and type $\theta_i\in\Theta_i$ the probability with which $r$ is played under mixed strategy $\sigma_i(\theta_i)$. Denote this probability, with slight abuse of notation, by $\sigma_i(\theta_i,r)$.


For any type profile $(\theta_{i_1},...,\theta_{i_n})$ (more precisely, for the priority vectors implied by the type profile) and any pure strategy vector $(r_{i_1},...,r_{i_n})\in\mathfrak{S}(\mathcal{S})^n$ the rules of the mechanism calculate the match outcome for each student. Denote the corresponding match outcome for student $i$ with $s_i((\theta_{i_1},...,\theta_{i_n}),(r_{i_1},...,r_{i_n}))$. Hence, for each student $i$ of type $\theta_i$ submitting ROL $r_i$, her belief about $\theta_{-i}$ conditional on $\theta_i$ leads to a belief about match probabilities $(f_{i,s})$ for each mixed strategy vector profile of all other students. Formally, for given $\theta_i$ and $r_i$, 
{\color{black}
\begin{equation}
    f_{i,s}=\sum_{\widehat \theta_{-i}\in\Theta_{-i}}\sum_{r_{-i}\in  \mathfrak{S}(\mathcal{S})^{n-1}}\text{Prob}(\theta_{-i}=\widehat \theta_{-i})|\theta_i)\cdot\prod_{j\in\mathcal I\setminus \{i\}}\sigma_j(\widehat \theta_j,r_j)\cdot \mathbbm{1}_{s=s_i((\theta_i,\widehat \theta_{-i}),(r_{i_1},...,r_{i_n}))}.
\end{equation} }

Evidently, match probabilities $f_{i,s}$ are continuous in each component $(\theta_j,r_j)$. Since, by \eqref{eq:payoff-distance}, student $i$'s utility is continuous in the match probabilities, it follows that each student $i$'s expected utility $\hat U_i(\theta_i,r_i,\sigma_{-i})\equiv U_i(\theta_i,(f_{i,s}(\theta_i,r_i,\sigma_{-i}))_{s\in\mathcal S})$ is continuous in each $\theta_i$, $r_i$, and each component of $\sigma_{-i}$. 

For a mixed strategy profile $\sigma=(\sigma_{i_1},...,\sigma_{i_n})$, define the function\footnote{\textcolor{black}{In contrast to the case with standard preferences, this function is not the expected utility of the match distribution from the mixed strategy $\sigma_i(\theta_i)$. With our interpretation of mixed strategies, it can be interpreted as the average utility an outsider would assign to student $i$ if the outsider assigns probability $\sigma_i(\theta_i,r_i)$ to the event that the student picks the pure strategy $r_i$.}}
\[p_i(\sigma,\theta_i)=\sum_{r_i\in \mathfrak{S}(\mathcal{S})} \sigma_i(\theta_i,r_i)\hat U_i(\theta_i,r_i,\sigma_{-i}).\]

Next, define for each $i$, $\theta_i$, and $r$ the function
\[\varphi_i(r,\theta_i,\sigma)=\max\{0,\hat U_i(\theta_i,r,\sigma_{-i})-p_i(\sigma,\theta_i)\}.\] 
This function is continuous in $\sigma$. Hence, the function $T$ on the set of strategy profiles defined by $T(\sigma)=\sigma'$, where for each component
\begin{equation}\label{Nash}
    \sigma_i'(\theta_i,r)=\frac{\sigma_i(\theta_i,r)+\varphi_i(r,\theta_i,\sigma)}{1+\sum_{r\in \mathfrak{S}(\mathcal{S})}\varphi_i(r,\theta_i,\sigma)},
\end{equation}
is a continuous function from a compact subspace of a finite dimensional vector space to itself.\footnote{\textcolor{black}{Intuitively, this function maps a strategy profile $\sigma$ to a another strategy profile $\sigma'$, \textcolor{black}{which for each agent increases the probability with which actions are played that increase the average utility.}}} Hence, by Brower's fixed-point theorem, it has a fixed point $\overline\sigma$. For each $i$ and $\theta_i$, let $r(i,\theta_i)$ be a ROL that minimizes $\hat U(\theta_i,r,\overline\sigma_{-i})$ among those ROLs played with positive probability under $\overline\sigma_i(\theta_i)$. Hence, $\varphi_i(r(i,\theta_i),\theta_i,\overline\sigma)=0$. Then, since $\overline\sigma$ is a fixed point, the denominator in (\ref{Nash}) for $\sigma=\overline\sigma$ must be one, and hence $\varphi_i(r,\theta_i,\overline\sigma)=0$ for all $i,\theta_i,$ and $r$. It follows that for all $i$ and $\theta_i$ all pure strategies $r_i$ that are played with positive probability under $\overline\sigma_i(\theta_i)$ induce the same utility $\hat U_i(\theta_i,r_i,\overline\sigma)$, and no other pure strategy action induces strictly higher utility. Hence, given $\overline\sigma_{-i}$ and $\theta_i$, any strategy that has positive probability in $\overline\sigma_i(\theta_i)$ is as CPE for player $i$, and the strategy profile $\overline\sigma$ indeed constitutes a CBNE.
\end{proof}
}

\begin{proof}[Proof of Lemma \ref{lem:elite}]
By \eqref{eq:payoff-distance},
an ROL which lists any subset of elite schools before the outside option induces an expected utility of $fv-\Lambda f(1-f)v$, where $f$ is the probability that at least one elite school of the subset is attainable. Since the utility is a convex function in $f$, it is maximized by either maximizing or minimizing $f$. Hence, by either listing all or none of the elite schools before the outside option.
\end{proof}

\begin{proof}[Proof of Lemma \ref{lem:CBNE_elite}]
The probability that there are less than 
capacity $q$ students with a score above $\omega$ among $(n-1)$ students,
\eq{ \label{eq:order-stat}
P_{n-1:q}(\omega):=\sum_{k=0}^{q-1}{n-1 \choose k} (1- G(\omega))^k G(\omega)^{n-1-k},
}
is continuously and monotonically increasing in $\omega$ from 0 to 1. 
Thus, there is a unique $\overline \omega^{l}$
such that $P_{n-1:q}(\overline{\omega}^l) = 1-\nicefrac{1}{\Lambda^l} \in (0,1)$.
Because $f(\omega)$ is minimal when all other students choose to apply, $f(\omega) \geq P_{n-1:q}(\omega)$
for all $\omega$ and all reporting strategies of the other students.
Hence, for any $\Lambda\leq\Lambda^l$ and any $\omega\geq\overline{\omega}^l$, 
we have $f(\omega) \geq P_{n-1:q}(\omega) \geq 1-\nicefrac{1}{\Lambda}$, 
meaning that applying to the elite school is a best response for all such types, as \eqref{eq:elite} holds.
  
Knowing that all students of score $\omega\geq \overline{\omega}^l$ apply, 
a student of type $\Lambda^l$ infers that for score $\omega\geq \overline{\omega}^l$ she has attainability probability $f(\omega)=P_{n-1:q}(\omega)$, and by construction applies if and only if her score satisfies 
$\omega\geq \overline{\omega}^l$. 

Next, because $\Lambda^{l-1}<\Lambda^l$, there are types $\omega<\overline{\omega}^l$ 
who prefer to apply as well.
A student of score $\omega<\overline{\omega}^l$ and sufficiently close to $\overline{\omega}^l$ 
expects acceptance if there are less than $q$ other students with  
a score either above $\overline{\omega}^l$ or a score in $[\omega,\overline{\omega}^l]$ and $\Lambda < \Lambda^l$. 
Again, this attainability probability is strictly and continuously increasing in $\omega$ 
which implies a unique cutoff 
$\overline{\omega}^{l-1}$ such that $f(\overline{\omega}^{l-1})=1-\nicefrac{1}{\Lambda^{l-1}}$. 
Hence, truthful reporting for type $\Lambda^{l-1}$ is optimal if and only if 
$\omega>\overline{\omega}^{l-1}$. Proceeding with this manner iteratively, 
we obtain an essentially unique choice-acclimating Bayesian Nash equilibrium.
\end{proof}

\begin{proof}[Proof of Lemma \ref{no_static}]
If DA is truthful, it is a static mechanism that 
implements the student-optimal stable allocations for all realizations of preferences. 
For the converse, take any static mechanism $(R,o)$ 
that implements the student-optimal outcome 
as CBNE for all realizations.  
More precisely, for each student $i$ there exists a strategy
$\sigma_i:\mathfrak{S}(\mathcal{S})\rightarrow R_i$ such that the joint strategy profile is a CBNE given $o$.
Consequently, the associated direct mechanism 
$\Big( \prod \mathfrak{S}(\mathcal S), o\circ (\sigma_1,...,\sigma_n) \Big)$ 
has a truthful CBNE by construction, and implements the student-optimal stable allocation. This direct mechanism asks students for their type vector and implements the student-optimal stable match (hence the DA match outcome) based on the ordinal preferences of their cardinal utility vector. As in this truthful equilibrium students have no incentive to misrepresent their cardinal utility vector $\mathbf v_i$ they have a forteriori no incentive to misrepresent in DA, which only asks for ordinal preferences to implement the student-optimal stable match. Hence, DA is truthful. 
\end{proof}

\bibliographystyle{elsarticle-harv}
{\singlespacing 
\small
\bibliography{lossaversion-bib}
}

\begin{landscape}
\thispagestyle{empty}
\subsection{Data from \citet{li2017}}

\begin{table}[h!]
\begin{adjustbox}{center,scale=0.9}
\begin{tabular}{|l|llllllllllllllllllll|ll|}
\cline{1-23}
        & \multicolumn{22}{c|}{PRIORITY SCORES}                                                                                                                                                                                                                                                                                                    \\ \hline
ROLs    & \multicolumn{2}{c|}{1}        & \multicolumn{2}{c|}{2}      & \multicolumn{2}{c|}{3}      & \multicolumn{2}{c|}{4}      & \multicolumn{2}{c|}{5}       & \multicolumn{2}{c|}{6}      & \multicolumn{2}{c|}{7}       & \multicolumn{2}{c|}{8}      & \multicolumn{2}{c|}{9}      & \multicolumn{2}{c|}{10} & \multicolumn{2}{c|}{ALL}     \\ \hline
\rowcolor[HTML]{9AFF99}
1234    & 55          & 61.1\%          & 48         & 57.1\%         & 47         & 58.8\%         & 42         & 67.7\%         & 32         & 55.2\%          & 49         & 79.0\%         & 58         & 74.4\%          & 48         & 85.7\%         & 59         & 84.3\%         & 73       & 91.3\%       & 511         & 71.0\%         \\
1243    & 1           & 1.1\%           & 1          & 1.2\%          & 1          & 1.3\%          & 0          & 0.0\%          & 0          & 0.0\%           & 1          & 1.6\%          & 1          & 1.3\%           & 0          & 0.0\%          & 1          & 1.4\%          & 0        & 0.0\%        & 6           & 0.8\%          \\
1324    & 2           & 2.2\%           & 3          & 3.6\%          & 2          & 2.5\%          & 1          & 1.6\%          & 2          & 3.4\%           & 0          & 0.0\%          & 1          & 1.3\%           & 0          & 0.0\%          & 1          & 1.4\%          & 0        & 0.0\%        & 12          & 1.7\%          \\
1342    & 1           & 1.1\%           & 0          & 0.0\%          & 0          & 0.0\%          & 0          & 0.0\%          & 1          & 1.7\%           & 0          & 0.0\%          & 0          & 0.0\%           & 0          & 0.0\%          & 0          & 0.0\%          & 1        & 1.3\%        & 3           & 0.4\%          \\
1423    & 0           & 0.0\%           & 1          & 1.2\%          & 0          & 0.0\%          & 1          & 1.6\%          & 1          & 1.7\%           & 0          & 0.0\%          & 0          & 0.0\%           & 0          & 0.0\%          & 0          & 0.0\%          & 0        & 0.0\%        & 3           & 0.4\%          \\
1432    & 0           & 0.0\%           & 0          & 0.0\%          & 0          & 0.0\%          & 0          & 0.0\%          & 0          & 0.0\%           & 0          & 0.0\%          & 0          & 0.0\%           & 0          & 0.0\%          & 0          & 0.0\%          & 1        & 1.3\%        & 1           & 0.1\%          \\
\rowcolor[HTML]{9AFF99}
2134    & 1           & 1.1\%           & 1          & 1.2\%          & 3          & 3.8\%          & \textbf{4} & \textbf{6.5\%} & \textbf{7} & \textbf{12.1\%} & \textbf{5} & \textbf{8.1\%} & \textbf{8} & \textbf{10.3\%} & \textbf{4} & \textbf{7.1\%} & \textbf{4} & \textbf{5.7\%} & 1        & 1.3\%        & \textbf{38} & \textbf{5.3\%} \\
2143    & 0           & 0.0\%           & 1          & 1.2\%          & 3          & 3.8\%          & 0          & 0.0\%          & 1          & 1.7\%           & 0          & 0.0\%          & 0          & 0.0\%           & 0          & 0.0\%          & 0          & 0.0\%          & 1        & 1.3\%        & 6           & 0.8\%          \\
\rowcolor[HTML]{9AFF99}
2314    & 1           & 1.1\%           & 2          & 2.4\%          & 2          & 2.5\%          & 1          & 1.6\%          & 2          & 3.4\%           & 1          & 1.6\%          & 0          & 0.0\%           & 2          & 3.6\%          & 1          & 1.4\%          & 1        & 1.3\%        & 13          & 1.8\%          \\
\rowcolor[HTML]{9AFF99}
2341    & 0           & 0.0\%           & 0          & 0.0\%          & 0          & 0.0\%          & 2          & 3.2\%          & 0          & 0.0\%           & 0          & 0.0\%          & 0          & 0.0\%           & 0          & 0.0\%          & 1          & 1.4\%          & 0        & 0.0\%        & 3           & 0.4\%          \\
2413    & 0           & 0.0\%           & 1          & 1.2\%          & 2          & 2.5\%          & 0          & 0.0\%          & 0          & 0.0\%           & 0          & 0.0\%          & 2          & 2.6\%           & 0          & 0.0\%          & 0          & 0.0\%          & 0        & 0.0\%        & 5           & 0.7\%          \\
2431    & 0           & 0.0\%           & 0          & 0.0\%          & 0          & 0.0\%          & 0          & 0.0\%          & 0          & 0.0\%           & 1          & 1.6\%          & 0          & 0.0\%           & 0          & 0.0\%          & 0          & 0.0\%          & 0        & 0.0\%        & 1           & 0.1\%          \\
3124    & 1           & 1.1\%           & 2          & 2.4\%          & 2          & 2.5\%          & 1          & 1.6\%          & 3          & 5.2\%           & 0          & 0.0\%          & 4          & 5.1\%           & 0          & 0.0\%          & 1          & 1.4\%          & 0        & 0.0\%        & 14          & 1.9\%          \\
3142    & 0           & 0.0\%           & 0          & 0.0\%          & 0          & 0.0\%          & 0          & 0.0\%          & 0          & 0.0\%           & 0          & 0.0\%          & 0          & 0.0\%           & 0          & 0.0\%          & 0          & 0.0\%          & 0        & 0.0\%        & 0           & 0.0\%          \\
\rowcolor[HTML]{9AFF99}
3214    & 6           & 6.7\%           & 5          & 6.0\%          & \textbf{6} & \textbf{7.5\%} & 3          & 4.8\%          & 2          & 3.4\%           & 0          & 0.0\%          & 0          & 0.0\%           & 1          & 1.8\%          & 0          & 0.0\%          & 0        & 0.0\%        & 23          & 3.2\%          \\
\rowcolor[HTML]{9AFF99}
3241    & 0           & 0.0\%           & 0          & 0.0\%          & 1          & 1.3\%          & 0          & 0.0\%          & 0          & 0.0\%           & 0          & 0.0\%          & 1          & 1.3\%           & 0          & 0.0\%          & 0          & 0.0\%          & 0        & 0.0\%        & 2           & 0.3\%          \\
3412    & 0           & 0.0\%           & 0          & 0.0\%          & 1          & 1.3\%          & 0          & 0.0\%          & 2          & 3.4\%           & 0          & 0.0\%          & 0          & 0.0\%           & 0          & 0.0\%          & 0          & 0.0\%          & 0        & 0.0\%        & 3           & 0.4\%          \\
\rowcolor[HTML]{9AFF99}
3421    & 3           & 3.3\%           & 2          & 2.4\%          & 0          & 0.0\%          & 0          & 0.0\%          & 0          & 0.0\%           & 0          & 0.0\%          & 0          & 0.0\%           & 0          & 0.0\%          & 0          & 0.0\%          & 0        & 0.0\%        & 5           & 0.7\%          \\
4123    & 1           & 1.1\%           & 2          & 2.4\%          & 1          & 1.3\%          & 0          & 0.0\%          & 1          & 1.7\%           & 2          & 3.2\%          & 1          & 1.3\%           & 0          & 0.0\%          & 0          & 0.0\%          & 1        & 1.3\%        & 9           & 1.3\%          \\
4132    & 0           & 0.0\%           & 1          & 1.2\%          & 0          & 0.0\%          & 0          & 0.0\%          & 0          & 0.0\%           & 0          & 0.0\%          & 0          & 0.0\%           & 0          & 0.0\%          & 0          & 0.0\%          & 0        & 0.0\%        & 1           & 0.1\%          \\
4213    & 1           & 1.1\%           & 1          & 1.2\%          & 0          & 0.0\%          & 1          & 1.6\%          & 3          & 5.2\%           & 1          & 1.6\%          & 1          & 1.3\%           & 0          & 0.0\%          & 0          & 0.0\%          & 0        & 0.0\%        & 8           & 1.1\%          \\
4231    & 1           & 1.1\%           & 2          & 2.4\%          & 2          & 2.5\%          & 0          & 0.0\%          & 0          & 0.0\%           & 0          & 0.0\%          & 0          & 0.0\%           & 1          & 1.8\%          & 0          & 0.0\%          & 0        & 0.0\%        & 6           & 0.8\%          \\
4312    & 0           & 0.0\%           & 4          & 4.8\%          & 4          & 5.0\%          & 3          & 4.8\%          & 0          & 0.0\%           & 0          & 0.0\%          & 0          & 0.0\%           & 0          & 0.0\%          & 0          & 0.0\%          & 1        & 1.3\%        & 12          & 1.7\%          \\
\rowcolor[HTML]{9AFF99}
4321    & \textbf{16} & \textbf{17.8\%} & \textbf{7} & \textbf{8.3\%} & 3          & 3.8\%          & 3          & 4.8\%          & 1          & 1.7\%           & 2          & 3.2\%          & 1          & 1.3\%           & 0          & 0.0\%          & 2          & 2.9\%          & 0        & 0.0\%        & 35          & 4.9\%          \\ \hline
Total   & 90          & 100.0\%         & 84         & 100.0\%        & 80         & 100.0\%        & 62         & 100.0\%        & 58         & 100.0\%         & 62         & 100.0\%        & 78         & 100.0\%         & 56         & 100.0\%        & 70         & 100.0\%        & 80       & 100.0\%      & 720         & 100.0\%        \\ \hline
misrep' & 35          & 38.9\%          & 36         & 42.9\%         & 33         & 41.3\%         & 20         & 32.3\%         & 26         & 44.8\%          & 13         & 21.0\%         & 20         & 25.6\%          & 8          & 14.3\%         & 11         & 15.7\%         & 7        & 8.8\%        & 209         & 29.0\%         \\
\rowcolor[HTML]{9AFF99} T\textcolor{black}{R}M      & 82          & 91.1\%          & 65         & 77.4\%         & 62         & 77.5\%         & 55         & 88.7\%         & 44         & 75.9\%          & 57         & 91.9\%         & 68         & 87.2\%          & 55         & 98.2\%         & 67         & 95.7\%         & 75       & 93.8\%       & 630         & 87.5\%         \\ \hline
\end{tabular}
\end{adjustbox}
\caption{Absolute and relative frequency of all ROLs for each priority score in the experiment by \citet{li2017}.
The top-\textcolor{black}{rank} monotone ROLs are marked, and the frequencies of the most common misrepresentations for each priority score are in bold.
}

\label{tab:li2}
\end{table}
\end{landscape}

\end{document}